\documentclass[11pt]{article}

\usepackage{array}
\usepackage{rotating}
\usepackage{graphicx}
\usepackage{amsmath, amssymb, amsthm}
\usepackage{booktabs}
\usepackage{url}
\usepackage{xcolor}
\usepackage{hyperref}
\usepackage{tikz}
\usetikzlibrary{arrows.meta, positioning, calc,decorations.markings}
\usepackage{longtable}
\hypersetup{
colorlinks=true,
linkcolor=blue,
citecolor=blue,
urlcolor=blue
}
\usepackage[outline]{contour}
\contourlength{1.5pt}
\usepackage[margin=1in]{geometry}
\usepackage{array}
\usepackage{rotating}
\usepackage{graphicx}
\usepackage{amsmath, amssymb, amsthm}
\usepackage{booktabs}
\usepackage{url}
\usepackage{xcolor}
\usepackage{hyperref}
\usepackage{tikz}
\usepackage{float}
\usetikzlibrary{arrows.meta, positioning, calc,decorations.markings}
\usepackage{longtable}
\usepackage[backend=biber,style=numeric,sorting=none]{biblatex}
\newcommand{\rev}[1]{#1}

\newtheorem{theorem}{Theorem}[section]
\newtheorem{proposition}[theorem]{Proposition}

\newcommand{\NMA}{\mathrm{NMA}}

\newcommand{\Var}{\mathrm{Var}}
\newcommand{\Cov}{\mathrm{Cov}}

\title{Contrast-Space Projection for Network Meta-Analysis:
An Exact and Invariant Study-Based Decomposition of Direct and Indirect Contributions}
\author{
Chong Wang$^{1,2}$\thanks{Corresponding author: chwang@iastate.edu} \and
Yanqi Zhang$^{1}$ \and
Zhezhen Jin$^{3}$ \and
Annette O'Connor$^{4}$
}

\date{}

\begin{document}
\maketitle

\begin{center}
{\small
$^1$ Department of Statistics, Iowa State University, USA \\
$^2$ Department of Veterinary Diagnostic and Production Animal Medicine, Iowa State University, USA \\
$^3$ Department of Biostatistics, Columbia University, USA \\
$^4$ Department of Large Animal Clinical Sciences, Michigan State University, USA
}
\end{center}

\begin{abstract}
Network meta-analysis (NMA) combines direct and indirect comparisons across a connected treatment network to estimate relative treatment effects. However, there is a lack of exact contribution decompositions that reproduce NMA estimates, particularly in the presence of multi-arm trials that induce within-study correlations.

We address this reproducibility gap by developing a contrast-space projection formulation of NMA. \rev{We express the NMA estimator as a linear mapping of observed pairwise contrasts onto the consistency-constrained contrast space. The same projection separates direct and indirect evidence, removes within-study representational redundancy through canonicalization, and, under fixed effects, represents the standard generalized Cochran \(Q\) decomposition into within-design heterogeneity and between-design inconsistency.}

Building on this representation, we introduce a rigorous study-based definition of direct and indirect evidence through a canonical within-study reduction that removes algebraic redundancy and yields a unique, invariant decomposition. This leads to exact covariance-aware decompositions of the NMA estimator into study-level direct and indirect contributions, with indirect evidence further resolved into path-level components. The resulting weights are directly analogous to inverse-variance weights in pairwise meta-analysis and enable, to our knowledge, the first forest-plot representation that exactly reconstructs the NMA estimator. \rev{Thus, the proposed decomposition remains algebraically tied to the single fitted NMA model rather than producing an alternative set of network estimates.}

The framework also yields projection-based diagnostic and graphical tools, including forest plots, tension plots, and path-based visualizations. Applications to empirical datasets demonstrate how the proposed approach provides a reproducible and interpretable framework for understanding evidence contributions in NMA, supporting transparent interpretation and reporting.
\end{abstract}


\section{Introduction}\label{sec:introduction}

Network meta-analysis (NMA) extends classical pairwise meta-analysis
by synthesizing evidence across multiple treatments within a connected
network of randomized trials and is now widely used in quantitative
evidence synthesis, clinical guidelines, and health-technology
assessment \parencite{lu2011linear, Salanti2011, Dias2013,  RuckerSchwarzer2015}. Under the consistency assumption, NMA combines direct and
indirect comparisons to estimate relative treatment effects,
including comparisons that have not been evaluated directly within any
single study. The resulting network estimates reflect a structured propagation of
evidence from observed treatment contrasts through the treatment
network, yielding a network-level pattern of evidence contributions. Understanding and decomposing this evidence is therefore central to the transparent interpretation of
NMA results. However, how such evidence is partitioned into direct and indirect components—and attributed to individual studies—remains poorly defined in current practice.

In traditional pairwise meta-analysis, the pooled treatment effect is typically expressed as a weighted average of study-level estimates. Because the estimator is linear, the relative contribution of each study can be calculated directly from the inverse-variance weights and is routinely displayed in forest plots. This transparent decomposition plays an important role in evidence interpretation, allowing systematic reviewers to identify which studies contribute most strongly to the pooled estimate \parencite{hedges1985statistical}. 

In contrast, such a transparent study-level decomposition is much less straightforward in NMA. NMA estimates combine multiple direct and indirect evidence pathways, often involving correlated contrasts arising from multi-arm trials. As a consequence, the algebraic relationship between individual studies and the final network estimates is not transparent. Although the statistical estimation of NMA effects is well understood,
the explicit and reproducible representation of network estimates as linear functions of the observed direct evidence, while preserving numerical reproducibility and compatibility with multi-arm covariance structures, remains largely underdeveloped. As NMA methods have become cornerstone tools in quantitative evidence synthesis, ensuring transparency in how evidence contributes to each estimate is critical. To support transparent reporting, as emphasized in the PRISMA-NMA (Preferred Reporting Items for Systematic Reviews and Meta-Analyses for Network Meta-Analyses) guidelines \parencite{Hutton2015}, it is critical to understand the exact contribution of each evidence source to the final network estimates. However, fulfilling this requirement in an exact and reproducible manner remains a significant methodological challenge, particularly in the presence of multi-arm trials. Existing contribution analyses are often constructed from marginal pairwise summaries, which do not preserve the full covariance structure used in the generalized least squares (GLS) formulation of NMA. As a result, such approaches may fail to reproduce the reported NMA estimates exactly and may yield decompositions that lack invariance to within-study parameterization.

Understanding how evidence propagates through treatment networks has attracted substantial methodological interest. \textcite{Rucker2012} established connections between NMA and electrical networks, interpreting treatment comparisons through graph Laplacians. Building on this idea, \textcite{Koenig2013} showed that linear operators arising in NMA can be interpreted as describing how evidence propagates between treatments. \textcite{Papakonstantinou2018} further proposed a stream decomposition in which edge-level flows are partitioned into paths connecting treatments, while \textcite{Davies2022} introduced a random-walk interpretation in an aggregate-level NMA formulation based on adjusted weights and a network of two-arm trials, deriving analytic expressions for evidence-flow proportions. In addition to flow-based interpretations, a separate line of work evaluates evidence importance through perturbation approaches such as leave-one-out (LOO) analyses. These methods quantify how removing a study or comparison changes a network estimate or its uncertainty. For example, \textcite{Rucker2020} formalize a notion of statistical importance for treatment-level NMA estimates, and \textcite{Mao2025} extend these ideas to component NMA. Such approaches measure the sensitivity of the estimator but do not provide an additive decomposition of the estimator itself. 

Despite these advances, an important reproducibility issue remains in current practice. Many contribution analyses construct the direct-evidence vector from marginal pairwise meta-analyses of each comparison. However, NMA estimators are derived from a \rev{GLS} system that explicitly incorporates correlations induced by multi-arm trials. When marginal pairwise summaries are substituted into such linear decompositions, the full study-level covariance structure is replaced by a comparison-level approximation, and the NMA formulation is replaced by a system of weight-adjusted two-arm trials. As a result, the decomposition may fail to reproduce the reported NMA estimates exactly. 
\rev{This limitation arises because existing approaches rely on approximations or transformations of the underlying GLS system, rather than operating directly on the GLS mapping that defines the estimator.} A formulation that works directly with this projection and preserves the full covariance structure is therefore needed.

A related and largely unresolved issue is that, under the current NMA estimation framework, there is no rigorous, model-based and explicit definition that identifies which component of the estimator corresponds to direct evidence at the study level. \rev{A trial provides direct evidence for a target comparison whenever it includes both treatments, regardless of how its within-study contrasts are parameterized. In a multi-arm trial, however, different parameterizations may or may not retain the target comparison explicitly, even though they represent the same study data and yield the same NMA estimate. Consequently, the same study contribution may be allocated differently between explicit direct and indirect components.} This lack of a rigorous definition and invariance complicates interpretation and limits the reproducibility of study-level contribution analyses.

Motivated by the lack of a rigorous and explicit model-based definition of direct and indirect evidence in NMA, as well as the need for exact algebraic reproducibility, we develop a contrast-space projection (CSP) formulation of NMA based on the Fisher-information geometry of the treatment network. Working directly in the space of all pairwise contrasts, we show that the NMA estimator can be expressed as a linear mapping of the observed evidence onto the consistency-constrained contrast space induced by orthogonal projection. This projection representation provides an exact and covariance-aware mapping between observed contrasts and network estimates, fully accounting for the dependence structure induced by multi-arm trials. Building on this formulation, we introduce a rigorous and explicit formulation of direct and indirect evidence and derive a representation-invariant decomposition of the network estimator into study-level components. \rev{This common projection formulation provides the foundation for the study-level decomposition and diagnostic tools developed below.}

\rev{A central principle of the framework is numerical and algebraic compatibility with the fitted NMA model. CSP does not generate a second set of network treatment-effect estimates: for every target comparison, its direct, indirect, study-level, and path-level contributions reconstruct the single fitted NMA estimate exactly. For the fixed-effects diagnostic decomposition, the global generalized Cochran statistic is partitioned exactly as \(Q_{\mathrm{net}}=Q_{\mathrm{het}}+Q_{\mathrm{inc}}\).} 

The main methodological contributions of this work are fourfold:

\begin{itemize}
  \item \rev{We place established GLS, hat-matrix, and graph-theoretic representations of NMA in a common full-contrast-space formulation that provides the foundation for the canonical study-based decomposition developed in Section~\ref{sec:canonical_decomposition}.}
  
  \item We introduce a rigorous study-based definition of direct and indirect evidence through a canonical within-study reduction that removes algebraic redundancy. This definition is unique and invariant to within-study parameterization.
  
  \item We develop exact, covariance-aware decompositions of the NMA estimator into study-level direct and indirect contributions, and further decompose indirect evidence into path-level components, with weights analogous to inverse-variance weights in pairwise meta-analysis.
  
\item Using these quantities, we construct practical diagnostic and graphical tools, including study-level forest plots, tension plots, path-based visualizations, and \rev{a projection representation of the established generalized Cochran \(Q\) decomposition \(Q_{\mathrm{net}}=Q_{\mathrm{het}}+Q_{\mathrm{inc}}\), separating within-design heterogeneity from between-design inconsistency \parencite{Krahn2013,Higgins2012,White2012}.}
\end{itemize}

\rev{By combining a projection-based formulation with a canonical, representation-invariant decomposition, the proposed approach bridges the gap between theoretical consistency, numerical reproducibility, and practical interpretability.}

The remainder of the paper is organized as follows.
Section~\ref{sec:formulation} introduces the contrast-space formulation
and projection representation of NMA.
Section~\ref{sec:canonical_decomposition} develops the canonical study-level decomposition and the representation-invariant definition of direct and indirect evidence, together with the induced path-based decomposition.
Section~\ref{sec:tools} presents the resulting diagnostic and graphical tools derived from the projection framework.
Section~\ref{sec:application} illustrates the methodology using several empirical NMA datasets.

\section{Methods}\label{sec:methods}

NMA combines observed direct contrasts to obtain treatment effect estimates that satisfy the consistency constraints of the treatment network, meaning that all pairwise comparisons are coherent with a single set of underlying treatment effects. This can be formulated as finding the closest set of consistency-compatible contrasts to the observed data, accounting for the covariance structure of the observations. \rev{The resulting covariance-weighted projection onto the consistency-compatible contrast space provides the organizing principle for the methodological developments that follow: it separates direct from indirect evidence, induces a canonical representation of study-level contributions, and provides a nested-projection representation of the standard generalized Cochran \(Q\) decomposition for heterogeneity and inconsistency.}

\subsection{Contrast-space formulation of NMA}\label{sec:formulation}

Standard formulations of NMA parameterize
treatment effects relative to a chosen reference treatment. In
contrast, we formulate NMA directly in the space of all estimable
pairwise treatment contrasts, which correspond to the comparisons
reported and interpreted in practice.

\rev{The principal ingredients in this subsection build on established NMA results: the linear hat-matrix representation and evidence-flow interpretation described by \textcite{Koenig2013}, the graph-Laplacian and effective-resistance geometry developed by \textcite{Rucker2012}, and the equivalence of reduced and full-pairwise representations of multi-arm trials studied by \textcite{Rucker2014reduce}. We collect these results in a common notation because they provide the foundation for the new canonical study-based decomposition introduced in Section~\ref{sec:canonical_decomposition}.}

This contrast-space formulation yields an exact linear mapping from
the observed contrasts to the fitted network estimates under the
\rev{GLS} model. The resulting representation
preserves the full covariance structure induced by multi-arm trials
and provides the foundation for a projection-based characterization of
evidence contributions.

\subsubsection{Contrast-space parameterization and consistency constraints}
\label{sec:csp}

We begin by parameterizing treatment effects in the full contrast space.
We assume the treatment network is connected and contains $T$ treatments. Let 
$m=\binom{T}{2}$ denote the number of pairwise treatment contrasts.
Define the contrast vector

\[
\boldsymbol\theta = (\theta_{12},\theta_{13},\dots,\theta_{(T-1)T})^\top 
\in \mathbb{R}^m ,
\]

where $\theta_{ab}$ denotes the treatment effect comparing treatment $b$ to treatment $a$.

Under the consistency assumption, treatment effects satisfy relations
such as $\theta_{ac} = \theta_{ab} + \theta_{bc}$. Collecting these
relations yields the linear constraint system

\[
\mathbf C \boldsymbol\theta = 0 ,
\]

where $\mathbf C$ encodes the consistency equations. The admissible
parameter space is therefore

\[
\mathcal M =
\{\boldsymbol\theta \in \mathbb{R}^m : \mathbf C \boldsymbol\theta = 0 \}.
\]

This consistency space has dimension $T-1$, corresponding to the number of independent treatment effects in a connected network.

\rev{Let $\mathbf B\in\mathbb R^{m\times(T-1)}$ be a full-column-rank
matrix whose columns form a basis for the consistency space
$\mathcal M=\ker(\mathbf C)$. Thus,}

\[
\rev{
\mathbf C\mathbf B=\mathbf 0,
\qquad
\operatorname{col}(\mathbf B)=\mathcal M.
}
\]

\rev{Every consistency-compatible contrast vector can therefore be
written uniquely as}

\[
\rev{
\boldsymbol\theta=\mathbf B\boldsymbol\beta,
\qquad
\boldsymbol\beta\in\mathbb R^{T-1}.
}
\]

\rev{The matrix $\mathbf B$ provides coordinates for the consistency
space. Although these coordinates depend on the selected basis, the
resulting fitted all-pairs contrast vector does not.}

The consistency constraints ensure that all pairwise treatment effects
can be represented through a set of underlying treatment-specific
parameters, yielding a coherent network structure. In particular, the
space $\mathcal M$ represents the set of contrast vectors that are
compatible with a common set of treatment effects, and thus corresponds
to the identifiable parameter space in a connected network. This
constraint structure plays a central role in defining the projection representation of the NMA estimator developed below.

\rev{\textbf{Running three-treatment example.} Consider treatments \(A\), \(B\), and \(C\), with \(\boldsymbol\theta=(\theta_{AB},\theta_{AC},\theta_{BC})^\top\). Taking \(\boldsymbol\beta=(\theta_{AB},\theta_{AC})^\top\), consistency implies \(\theta_{BC}=\theta_{AC}-\theta_{AB}\), so that}
\[
\rev{\boldsymbol\theta=\mathbf B\boldsymbol\beta,\qquad \mathbf B=\begin{pmatrix}1&0\\0&1\\-1&1\end{pmatrix}.}
\]

\subsubsection{Observed contrasts and covariance structure}

\rev{Let \(\mathbf y\in\mathbb R^{n_c}\) denote the vector of observed relative treatment effects, where \(n_c\) is the contrast-level sample size, defined as the total number of retained study-level contrasts in the analysis. This quantity is distinct from \(m=\binom{T}{2}\), the number of possible pairwise treatment contrasts among the \(T\) treatments: \(n_c\) depends on the included studies and their within-study representation, whereas \(m\) depends only on the number of treatments.} For studies with $k \geq 3$ treatment arms, the dimension and composition of $\mathbf y$ depend on how within-study correlations are represented \parencite{Rucker2014reduce}. 

Multi-arm trials may be represented either through a set of $k-1$ baseline contrasts or through the full set of $k(k-1)/2$ pairwise contrasts together with their induced covariance structure \parencite{Lu2004, Jackson2012, Rucker2014reduce}. 

\begin{enumerate}

\item \textbf{Dimension reduction (basic-contrast representation).}  
The standard approach selects one treatment as a baseline within each trial and constructs $\mathbf y$ from the resulting $k-1$ contrasts relative to that baseline \parencite{Lu2004, Higgins2012}. This representation ensures that the within-study covariance matrix is non-singular and leads to the familiar formulation used in most NMA implementations. The remaining pairwise contrasts within the trial are implicitly represented through linear combinations of the retained contrasts.

\item \textbf{Full pairwise within-study contrast representation.}  
An alternative representation retains all $k(k-1)/2$ pairwise contrasts that can be formed within each multi-arm study. In this case the within-study covariance matrix is singular because the contrasts are linearly dependent. Graph-theoretical and flow-based approaches often operate on this expanded set of within-study contrasts and adjust the covariance structure to account for the induced correlations \parencite{Rucker2014reduce, Papakonstantinou2018}. In the present framework, this representation is retained and the singular covariance matrix is handled directly using the Moore–Penrose pseudoinverse.
This representation provides a common contrast space in which all
within-study parameterizations can be embedded and will serve as the
foundation for the representation-invariant decomposition developed
in Section~\ref{sec:canonical_decomposition}.
\end{enumerate}

The graphical representation of NMA, in which treatments are nodes and pairwise comparisons form the edges of a network, corresponds naturally to the full pairwise contrast space. In this setting, a target network comparison can be expressed as a linear combination of the direct comparison estimates in the network \parencite{Papakonstantinou2018}. 

In the basic-contrast representation, the retained within-study contrasts depend strictly on the choice of baseline treatment. Even when a natural baseline exists—such as a placebo or standard care—any target comparison between two non-baseline treatments (e.g., comparing two active arms) will not appear explicitly in the observed contrast vector, and will be regarded as indirect evidence through this specific network structure. Furthermore, in many applications where no baseline is intrinsically preferred, different choices yield different but algebraically equivalent parameterizations of the exact same multi-arm study. As a result, whether a target comparison appears explicitly as a retained study contrast depends entirely on this representational choice, even though the resulting NMA estimator remains unchanged. This parameterization dependence is a primary source of ambiguity in study-level definitions of direct and indirect evidence and motivates the representation-invariant construction developed later.

Let $\mathbf V$ denote the \rev{$n_c \times n_c$} covariance matrix of $\mathbf y$, capturing both within-study correlations and potential between-study heterogeneity. The observed study-level contrast model is

\begin{equation}
\mathbf y = \mathbf X \boldsymbol\theta + \boldsymbol\varepsilon,
\qquad
\boldsymbol\varepsilon \sim N(0,\mathbf V),
\label{eq:contrast_model}
\end{equation}

where $\mathbf X$ denotes the \rev{GLS} design matrix for the observed study contrasts. Each row of $\mathbf X \in \mathbb{R}^{\rev{n_c} \times m}$ specifies the linear combination of treatment contrasts corresponding to a particular observation in $\mathbf y$. \rev{Throughout Section~\ref{sec:formulation}, \(\mathbf y\), \(\mathbf X\), and \(\mathbf V\) denote a generic valid representation of the observed contrasts, which may be either a reduced basic-contrast representation or a full pairwise representation. The tilde notation introduced later in Section~\ref{sec:canonical_decomposition} is reserved specifically for the transformation of the representation used for fitting into the common full pairwise within-study contrast space required for the canonical CSP decomposition.}

\rev{Under a fixed-effects specification, \(\mathbf V\) consists solely of within-study sampling covariance matrices and is block-diagonal across studies. Under a common-heterogeneity random-effects specification, the heterogeneity variance \(\tau^2\) is incorporated through an additional within-study contrast covariance block: each treatment contrast has heterogeneity variance \(\tau^2\), two contrasts from the same multi-arm study that share an arm have heterogeneity covariance \(\pm\tau^2/2\) according to their orientations, and contrasts involving disjoint pairs have zero heterogeneity covariance. The heterogeneity parameter \(\tau^2\) is first estimated, for example by restricted maximum likelihood, after which \(\widehat{\tau}^2\) is used to construct \(\mathbf V(\widehat{\tau}^2)\). Conditional on the resulting covariance matrix, the NMA projection and the CSP direct, indirect, study-level, and path-level decompositions proceed exactly as under fixed effects.}

\rev{For the running example, suppose Study 1 is a three-arm trial comparing \(A\), \(B\), and \(C\), represented initially by \(y_{AB}^{(1)}=1\) and \(y_{AC}^{(1)}=2\), with covariance matrix \(\mathbf V_1=\left(\begin{smallmatrix}2&1\\1&2\end{smallmatrix}\right)\). Study 2 reports \(y_{AB}^{(2)}=2\) with variance \(2\), and Study 3 reports \(y_{AC}^{(3)}=1\) with variance \(2\). Using the reduced representation and ordering \((y_{AB}^{(1)},y_{AC}^{(1)},y_{AB}^{(2)},y_{AC}^{(3)})\),}

\[
\rev{\mathbf y=\begin{pmatrix}1\\2\\2\\1\end{pmatrix},\qquad \mathbf X\mathbf B=\begin{pmatrix}1&0\\0&1\\1&0\\0&1\end{pmatrix},\qquad \mathbf V=\begin{pmatrix}2&1&0&0\\1&2&0&0\\0&0&2&0\\0&0&0&2\end{pmatrix}.}
\]

\rev{For Study 1, the corresponding full pairwise representation is obtained using}

\[
\rev{\mathbf A_1=\begin{pmatrix}1&0\\0&1\\-1&1\end{pmatrix},\qquad \tilde{\mathbf y}_1=\mathbf A_1\mathbf y_1=\begin{pmatrix}1\\2\\1\end{pmatrix},\qquad \tilde{\mathbf V}_1=\mathbf A_1\mathbf V_1\mathbf A_1^\top=\begin{pmatrix}2&1&-1\\1&2&1\\-1&1&2\end{pmatrix}.}
\]

\rev{The full-representation within-study covariance matrix \(\tilde{\mathbf V}_1\) has rank \(2\), illustrating the singular covariance structure handled by the Moore--Penrose pseudoinverse.}

\subsubsection{Projection representation of the NMA estimator}

The NMA estimator is obtained as the constrained \rev{GLS} solution
\begin{equation}
\hat{\boldsymbol\theta}^{\NMA}
=
\arg\min_{\boldsymbol\theta}
(\mathbf y-\mathbf X\boldsymbol\theta)^\top
\mathbf V^{+}
(\mathbf y-\mathbf X\boldsymbol\theta)
\quad
\text{subject to }
\mathbf C \boldsymbol\theta = 0 .
\label{eq:cgls}
\end{equation}
When $\mathbf V$ is singular, as in the full contrast representation,
we formulate the \rev{GLS} criterion using the
Moore--Penrose pseudoinverse $\mathbf V^{+}$. When $\mathbf V$ is
nonsingular, $\mathbf V^{+} = \mathbf V^{-1}$, and the formulation
reduces to the standard \rev{GLS} problem.

\rev{Using the consistency-space parameterization
$\boldsymbol\theta=\mathbf B\boldsymbol\beta$, the constrained problem
in Equation~\eqref{eq:cgls} is equivalent to}

\[
\rev{
\widehat{\boldsymbol\beta}
=
\arg\min_{\boldsymbol\beta}
(\mathbf y-\mathbf X\mathbf B\boldsymbol\beta)^\top
\mathbf V^+
(\mathbf y-\mathbf X\mathbf B\boldsymbol\beta).
}
\]

\rev{Define the Fisher information matrix in the
$(T-1)$-dimensional consistency coordinates as}

\[
\rev{
\mathcal I_{\mathbf B}
=
\mathbf B^\top\mathbf X^\top
\mathbf V^+\mathbf X\mathbf B.
}
\]

\rev{Under the connected-network identifiability conditions assumed
throughout, $\mathcal I_{\mathbf B}$ is nonsingular, and}

\[
\rev{
\widehat{\boldsymbol\beta}
=
\mathcal I_{\mathbf B}^{-1}
\mathbf B^\top\mathbf X^\top
\mathbf V^+\mathbf y.
}
\]

\rev{Consequently, the constrained NMA estimator has the explicit
linear representation}

\begin{equation}
\rev{
\widehat{\boldsymbol\theta}^{\NMA}
=
\mathbf P\mathbf y,
\qquad
\mathbf P
=
\mathbf B
\left(
\mathbf B^\top\mathbf X^\top
\mathbf V^+\mathbf X\mathbf B
\right)^{-1}
\mathbf B^\top\mathbf X^\top
\mathbf V^+.
}
\label{eq:P_main}
\end{equation}

\rev{Because $\mathbf C\mathbf B=\mathbf 0$, this mapping explicitly satisfies}

\[
\rev{
\mathbf C\widehat{\boldsymbol\theta}^{\NMA}
=
\mathbf C\mathbf P\mathbf y
=
\mathbf 0.
}
\]

\rev{The basis $\mathbf B$ is only a coordinate device. In particular,
if $\mathbf B_*=\mathbf B\mathbf R$ for a nonsingular
$(T-1)\times(T-1)$ matrix $\mathbf R$, then substituting
$\mathbf B_*$ into Equation~\eqref{eq:P_main} produces the same mapping
$\mathbf P$. Thus, the fitted all-pairs contrasts do not depend on the
choice of reference treatment or basis for $\mathcal M$.}

Although the reduced and full contrast representations differ algebraically, they correspond to the same underlying likelihood for the observed contrasts. When the exact within-study covariance structure is correctly specified, both representations yield identical NMA estimators for the treatment effects; the distinction lies only in how the observed contrasts and their covariance structure are expressed.

\begin{proposition}[Equivalence of reduced and full contrast representations]
\label{prop:equiv_embedding}

\rev{Within this proposition, let \((\mathbf y_b,\mathbf X_b,\mathbf V_b)\) denote a reduced basic-contrast representation, and let \((\mathbf y,\mathbf X,\mathbf V)\) denote its corresponding full contrast representation obtained through the linear embedding}
\[
\mathbf y = \mathbf A \mathbf y_b,\qquad
\mathbf X = \mathbf A \mathbf X_b,\qquad
\mathbf V = \mathbf A \mathbf V_b \mathbf A^\top,
\]
where $\mathbf A$ has full column rank and $\mathbf V_b$ is
nonsingular.

Then
\[
\mathbf X^\top \mathbf V^{+} \mathbf y
=
\mathbf X_b^\top \mathbf V_b^{-1} \mathbf y_b,
\qquad
\mathbf X^\top \mathbf V^{+} \mathbf X
=
\mathbf X_b^\top \mathbf V_b^{-1} \mathbf X_b.
\]

\rev{To verify these identities, let \(\mathbf G=\mathbf A\mathbf V_b^{1/2}\). Since \(\mathbf A\) has full column rank and \(\mathbf V_b\) is a nonsingular covariance matrix, \(\mathbf G\) has full column rank and \(\mathbf V=\mathbf G\mathbf G^\top\). Hence \(\mathbf V^+=\mathbf G^{+\top}\mathbf G^+\) and \(\mathbf G^+\mathbf G=\mathbf I\). Because \(\mathbf A=\mathbf G\mathbf V_b^{-1/2}\),}
\[
\rev{
\mathbf A^\top\mathbf V^+\mathbf A
=
\mathbf V_b^{-1/2}
\mathbf G^\top\mathbf G^{+\top}
\mathbf G^+\mathbf G
\mathbf V_b^{-1/2}
=
\mathbf V_b^{-1}.
}
\]
\rev{Substituting \(\mathbf y=\mathbf A\mathbf y_b\) and \(\mathbf X=\mathbf A\mathbf X_b\) then gives the two identities above.}

\rev{Premultiplying the first identity by $\mathbf B^\top$, and
premultiplying and postmultiplying the second identity by
$\mathbf B^\top$ and $\mathbf B$, respectively, gives}
\[
\rev{
\mathbf B^\top\mathbf X^\top\mathbf V^+\mathbf y
=
\mathbf B^\top\mathbf X_b^\top
\mathbf V_b^{-1}\mathbf y_b
}
\]
\rev{and}
\[
\rev{
\mathbf B^\top\mathbf X^\top\mathbf V^+\mathbf X\mathbf B
=
\mathbf B^\top\mathbf X_b^\top
\mathbf V_b^{-1}\mathbf X_b\mathbf B.
}
\]

\rev{Therefore, the constrained estimator in
Equation~\eqref{eq:P_main}, computed from the full contrast
representation using $\mathbf V^+$, coincides exactly with that
obtained from the reduced representation using
$\mathbf V_b^{-1}$.}

\end{proposition}

\medskip
\noindent
\rev{Proposition~\ref{prop:equiv_embedding} is an application of the standard invariance of GLS under a full-column-rank linear embedding of the observation space. Its purpose here is not to introduce a new invariance result, but to establish formally that the common full-contrast representation used in Section~\ref{sec:canonical_decomposition} preserves the fitted NMA estimator, score vector, and Fisher information.}

\rev{\textbf{Running example continued.} For the target comparison \(B{:}C\), the projection coefficients in the reduced representation are}

\[
\rev{\mathbf p_{BC}^\top=\begin{pmatrix}-\frac{2}{3}&\frac{2}{3}&-\frac{1}{3}&\frac{1}{3}\end{pmatrix},}
\]

\rev{and therefore}

\[
\rev{\widehat\theta_{BC}^{\NMA}=\mathbf p_{BC}^\top\mathbf y=-\frac{2}{3}(1)+\frac{2}{3}(2)-\frac{1}{3}(2)+\frac{1}{3}(1)=\frac{1}{3}.}
\]

\rev{Using all three pairwise contrasts for Study 1 and the pseudoinverse of the resulting singular covariance matrix gives}

\[
\rev{\tilde{\mathbf p}_{BC}^\top=\begin{pmatrix}-\frac{2}{9}&\frac{2}{9}&\frac{4}{9}&-\frac{1}{3}&\frac{1}{3}\end{pmatrix},}
\]

\rev{where the first three entries correspond to \((y_{AB}^{(1)},y_{AC}^{(1)},y_{BC}^{(1)})\). This representation also gives \(\tilde{\mathbf p}_{BC}^\top\tilde{\mathbf y}=1/3\), illustrating the exact equivalence of the reduced and full representations.}

\rev{This formulation admits a geometric interpretation in the
observation space. The fitted values are}

\[
\rev{
\widehat{\mathbf y}
=
\mathbf X\widehat{\boldsymbol\theta}^{\NMA}
=
\mathbf X\mathbf B\widehat{\boldsymbol\beta}.
}
\]

\rev{They represent the $\mathbf V^+$-orthogonal projection of
$\mathbf y$ onto the consistency-constrained fitted space
$\operatorname{col}(\mathbf X\mathbf B)$, rather than onto the larger
space $\operatorname{col}(\mathbf X)$. When $\mathbf V$ is singular,
this projection is understood on $\operatorname{col}(\mathbf V)$,
where the induced semi-inner product is nondegenerate.}

\rev{Equation~\eqref{eq:P_main} maps the observed contrast vector $\mathbf y\in\mathbb R^{n_c}$ to a fitted contrast vector $\widehat{\boldsymbol\theta}^{\NMA}\in\mathcal M$, while $\mathbf X\mathbf P$ is the corresponding projection operator in the observation space.}

\medskip
\noindent

\textbf{Inner product interpretation.}
Throughout, expressions of the form
\[
\langle \mathbf a, \mathbf b \rangle_{\mathbf V^{+}}
=
\mathbf a^\top \mathbf V^+ \mathbf b
\]
define a symmetric bilinear form on $\mathbb{R}^{n_c}$. When $\mathbf V$ is
nonsingular, this corresponds to a standard inner product. When
$\mathbf V$ is singular, the form is a semi-inner product on
$\mathbb{R}^{n_c}$, but becomes a valid inner product when restricted to
$\operatorname{col}(\mathbf V)$, which contains all admissible contrast
vectors arising from the model. All projection statements below are
understood with respect to this induced geometry on
$\operatorname{col}(\mathbf V)$.

\rev{Although the full within-study contrast representation increases the nominal dimension of the system, computation remains tractable because studies are independent and the covariance matrix is block-diagonal. For a study with \(t_k\) treatment arms, the full block contains \(m_k=\binom{t_k}{2}\) contrasts but has rank only \(t_k-1\); generalized inverses can therefore be computed study by study rather than by forming a dense pseudoinverse of the complete network covariance matrix. The subsequent canonical within-study reduction is also finite: each residual allocation exhausts at least one source or sink, requiring at most the bound given in Proposition~\ref{prop:canonical_reduction}. For the across-study path decomposition, each bottleneck extraction removes at least one positive study-labeled residual edge, so the number of extraction steps is at most the number of such edges. Path enumeration can nevertheless become increasingly costly and difficult to display in very dense networks, which motivates the aggregated study-level summaries used for the large antidepressant network.}

\subsubsection{Fisher-information geometry of the projection}

\rev{The relevant Fisher information matrix for the
consistency-constrained model is}

\[
\rev{
\mathcal I_{\mathbf B}
=
\mathbf B^\top\mathbf X^\top
\mathbf V^+\mathbf X\mathbf B.
}
\]

\rev{For the usual reference-treatment basis, this matrix has the
structure of a grounded weighted graph Laplacian. More generally, it
represents the same information operator in an arbitrary basis for the
consistency space.}

\rev{For a target comparison $a{:}b$, let
$\mathbf c_{ab}\in\mathbb R^m$ select the corresponding coordinate
from the full contrast vector, so that
$\theta_{ab}=\mathbf c_{ab}^\top\boldsymbol\theta$. Define}

\[
\rev{
\mathbf d_{ab}
=
\mathbf B^\top\mathbf c_{ab},
\qquad
\mathbf u_{ab}
=
\mathcal I_{\mathbf B}^{-1}\mathbf d_{ab}.
}
\]

\rev{The coefficient vector acting on the observed contrasts is then}

\[
\rev{
\mathbf p_{ab}
=
\mathbf V^+\mathbf X\mathbf B\mathbf u_{ab},
}
\]

\rev{and the target NMA estimate can be written as}

\[
\rev{
\widehat\theta_{ab}^{\NMA}
=
\mathbf p_{ab}^\top\mathbf y.
}
\]

\rev{This formulation has the usual minimum-energy electrical-network interpretation induced by the GLS precision operator \(\mathbf V^+\), equivalently with \(\mathbf V\) representing the covariance/resistance structure, linking the statistical model to classical graph-theoretic and electrical-network interpretations.} \rev{Importantly, this analogy should not be interpreted as a physical electrical flow. Before canonicalization, the raw full-contrast projection coefficients may contain local reverse directions or apparent cycles arising from algebraic redundancy in multi-arm contrast coordinates. These routes need not follow a common treatment-potential ordering. The canonical treatment-balance reduction removes this redundancy while preserving each study's induced linear functional. Under the standard arm-based covariance construction, Proposition~\ref{prop:acyclic_canonical_graph} shows that every canonical edge follows a common treatment-potential ordering and that the pooled canonical graph is acyclic.} The Fisher-information representation also induces an edge-level summary of how evidence for the target comparison $a{:}b$ is distributed across treatment comparisons. At the treatment level, this representation satisfies a conservation property: the total contribution flowing out of treatment $a$ equals one, the total contribution flowing into treatment $b$ equals one, and intermediate treatments have balanced incoming and outgoing contributions.

\rev{Conditional on the covariance matrix \(\mathbf V\), the mapping \(\mathbf P\) depends only on the design matrix \(\mathbf X\), the covariance structure \(\mathbf V\), and the consistency space \(\mathcal M\), and does not otherwise depend on the observed outcomes \(\mathbf y\). Under a random-effects specification in which \(\widehat{\tau}^2\) is estimated from the data, this statement is understood conditionally on the fitted covariance matrix \(\mathbf V(\widehat{\tau}^2)\). As
shown above, the mapping is invariant to the particular basis
$\mathbf B$ used to represent $\mathcal M$.} Each row of $\mathbf P$ describes how information is distributed from the observed contrasts to a target comparison, incorporating multi-arm correlation and, when present, random-effects heterogeneity through the Fisher information structure. In contrast, several existing proportion-contribution approaches apply additional transformations to projection- or hat-matrix representations to obtain nonnegative comparison-level weights. While such quantities may be useful descriptively, they no longer correspond directly to the underlying Fisher-informed linear operator. Working directly with $\mathbf P$ preserves the exact covariance structure and yields an algebraically exact and reproducible decomposition of the network estimator. While the projection-induced mapping \(\mathbf P\) is algebraically equivalent to the \rev{GLS} estimator, its role here as the primary object of analysis enables exact, covariance-aware decompositions and representation-invariant study-level contributions that are not available in standard formulations.

\rev{In particular, the variance of the fitted target comparison is}

\[
\rev{
\Var\!\left(\widehat\theta_{ab}^{\NMA}\right)
=
\mathbf d_{ab}^\top
\mathcal I_{\mathbf B}^{-1}
\mathbf d_{ab}
=
\mathbf c_{ab}^\top
\mathbf B
\mathcal I_{\mathbf B}^{-1}
\mathbf B^\top
\mathbf c_{ab}.
}
\]
\rev{For the standard treatment-reference basis, this quantity
coincides with the effective resistance between treatments $a$ and $b$
in the corresponding weighted network.}

\subsubsection{\rev{Generalization of pairwise inverse-variance weights}}

A hallmark of traditional pairwise meta-analysis is the transparency
provided by inverse-variance weights, typically displayed as
percentages in a forest plot to denote each study's contribution to
the pooled estimate \parencite{hedges1985statistical}. NMA
lacks an exact analogue because the estimator combines direct and
indirect evidence through a \rev{GLS} system that
incorporates multi-arm covariance.

The contrast-space projection framework provides a natural
generalization of these weights. Because the NMA estimator admits the
linear representation
\[
\hat{\boldsymbol\theta}^{\NMA} = \mathbf P \mathbf y,
\]
the coefficients $p_{ij}$ describe how each observed contrast $y_j$
contributes to the network estimate $\hat{\theta}_i^{\NMA}$.

In the special case of a degenerate network consisting of a single
pairwise comparison with $k$ independent studies, the design matrix
reduces to $\mathbf X=\mathbf 1$ and the covariance matrix $\mathbf V$
is diagonal. The \rev{GLS} estimator then becomes

\[
\hat{\theta}
=
\frac{\sum_{i=1}^k \sigma_i^{-2} y_i}{\sum_{i=1}^k \sigma_i^{-2}},
\]

which is the classical inverse-variance weighted mean. In this case the projection representation reduces exactly to the familiar relative weights $\sigma_i^{-2} / \sum_i \sigma_i^{-2}$.

\subsubsection{Structural direct and indirect decomposition}

Let $\mathbf p_{ab}^\top$ denote the row of the projection-induced mapping $\mathbf P$ corresponding to the target comparison $a{:}b$, and let
$j$ index the observed contrasts in the network. The \rev{NMA} estimator can be written as
\[
\hat\theta_{ab}^{\NMA} = \sum_j p_{ab,j} y_j .
\]
The projection coefficients may take negative values and are therefore interpreted as signed contributions rather than proportions.

To distinguish \emph{structural} direct and indirect contributions, let
$D_{ab}$ denote the set of indices corresponding to observed contrasts
that directly compare treatments $a$ and $b$. The estimator can then be
decomposed as
\[
\hat\theta_{ab}^{\NMA}
=
\sum_{j\in D_{ab}} p_{ab,j} y_j
+
\sum_{j\notin D_{ab}} p_{ab,j} y_j .
\]

The first term represents contributions from direct comparisons of
$a$ and $b$, while the second term aggregates contributions from all
remaining contrasts, corresponding to indirect evidence transmitted
through the network.

This decomposition follows directly from the projection representation
and provides an exact algebraic separation of direct and indirect
evidence at the level of the estimator. This decomposition is exact and model-based, requiring no approximation
or reparameterization. The definition is induced purely
by the structure of the projection-induced mapping $\mathbf P$ and depends only
on how observed contrasts are indexed in the network representation.
In particular, if the comparison $a{:}b$ is not observed directly in the
network, then $D_{ab}=\varnothing$ and the estimator is entirely
determined by indirect evidence. This structural definition aligns with formulations that express network
estimates as linear combinations of direct comparisons
\parencite{Papakonstantinou2018}, where the comparison $(a,b)$
corresponds to direct evidence and all other comparisons contribute
indirectly.

In many applications, however, direct and indirect evidence are defined
at the study level. For a target comparison \(a{:}b\), direct evidence
refers to information arising from studies that include both treatments
\(a\) and \(b\), whereas indirect evidence is obtained by combining
results across multiple studies and cannot be obtained from any single
study alone \parencite{Salanti2011, Dias2013, Hutton2015}. This distinction highlights two fundamentally different decompositions of the same projection representation: a structural decomposition defined at the level of treatment contrasts, which depends on how contrasts are represented, and a study-based decomposition that partitions contributions according to study membership and should be invariant to within-study parameterization. These two perspectives are not equivalent in general, particularly in the presence of multi-arm trials. While the structural decomposition is algebraically exact, it depends on the representation of the observed contrasts and is therefore not invariant to within-study parameterization, motivating the development of a rigorous formulation of study-level contributions within the contrast-space projection framework in the following subsection. \rev{The same study-based definition of direct and indirect evidence is used in the symmetric side-splitting model of \textcite{White2015} and the evidence-splitting approach of \textcite{ShihTu2021}. These methods were developed to assess inconsistency between direct and indirect evidence, rather than to decompose the fitted standard NMA estimate into direct and indirect contributions.}

\rev{The mapping $\mathbf P$, which defines the network estimator through its linear representation, does more than assign weights to observed contrasts.} Because it acts on the full contrast space, it provides the algebraic mechanism for separating structural direct from indirect evidence and for identifying representation-dependent ambiguity at the study level. \rev{The same projection geometry also provides, through nested fitted spaces, a direct representation of the established decomposition \(Q_{\mathrm{net}}=Q_{\mathrm{het}}+Q_{\mathrm{inc}}\).} The canonical construction developed in the following subsection builds directly on this projection geometry.

\subsection{Canonical study-based decomposition and path representation}
\label{sec:canonical_decomposition}

We develop a study-based, representation-invariant decomposition of the
NMA estimator into study-level contributions and
path-level components. As discussed above, while the projection
representation induces an exact contrast-based (structural)
decomposition at the level of treatment contrasts, this decomposition
is not invariant to within-study parameterization in the presence of
multi-arm trials. 

A central difficulty is that multi-arm trials admit multiple
algebraically equivalent representations, which can lead to different
allocations of a study’s contribution into direct and indirect
components. The canonical construction resolves this ambiguity by
removing within-study redundancy and defining a unique decomposition
that is invariant to the choice of within-study representation.

To motivate the need for a canonical decomposition, consider a simple
three-treatment network with treatments $A$, $B$, and $C$, and suppose
the target comparison is $B{:}C$. In the full pairwise contrast
representation, the observed data vector $\mathbf y$ may include the
direct comparison $B{:}C$ as well as contrasts such as $A{:}B$ and
$A{:}C$, yielding both a structural direct component and an indirect
pathway $B \leftarrow A \rightarrow C$. In contrast, under a baseline
parameterization (e.g., taking $A$ as reference), only the contrasts
$A{:}B$ and $A{:}C$ are represented explicitly, so the same network
estimate for $B{:}C$ appears to arise entirely through structural
indirect evidence, with no explicit direct component.

These representations are algebraically equivalent and yield identical
network estimates, but they induce different \emph{structural}
decompositions, since the notion of “direct” and “indirect” depends on
how contrasts are encoded in $\mathbf y$. Consequently, the resulting
decomposition of the estimator into direct and indirect contributions,
as well as any derived proportions or weights, can vary across
equivalent parameterizations.

In particular, the classification of a study contribution as “direct”
or “indirect” depends on how the within-study contrasts are
parameterized, including whether the target comparison appears
explicitly among the retained contrasts and how its contribution is
represented. Even when the target comparison appears explicitly among the retained contrasts, its contribution need not be fully represented within the direct component, as part of its effect may be redistributed through indirect pathways. Because this depends on the choice of within-study parameterization
used to represent multi-arm trials in standard NMA,
such a decomposition is not uniquely determined  and is therefore not invariant to the choice of within-study representation.

This structural notion differs from the \emph{study-based} interpretation,
in which direct evidence for $B{:}C$ arises solely from studies that
include both treatments $B$ and $C$, and indirect evidence arises from
combining information across multiple studies. Under this study-based
perspective, the decomposition—and in particular the relative magnitude
of direct and indirect contributions—should be invariant.

\rev{To resolve this issue, we develop a canonical decomposition that (i) defines study-level contributions that are invariant to within-study parameterization, (ii) removes algebraic redundancy through a unique within-study representation, and (iii) induces a well-defined network flow whose decomposition into individual paths is uniquely selected by the stated deterministic path-extraction rule.}

\rev{The construction proceeds in four steps. First, the projection
induces an exact decomposition of the network estimator into
study-level contributions. Second, each study-specific coefficient
functional is summarized by its treatment-level balance vector, which
is invariant to the representation of the within-study contrasts.
Third, the maximum feasible contribution to the target direct
comparison is extracted, and the remaining balances are assigned by a deterministic lexicographic max--min rule. Finally, the resulting study-labeled indirect edges are assembled across studies and decomposed into paths from the target source treatment to the target sink treatment.}

\rev{The study-based direct/indirect split itself is uniquely defined. For a target comparison \(a{:}b\), a study provides direct evidence when it contains both treatments \(a\) and \(b\); in the treatment-balance representation, the corresponding direct component is the uniquely determined maximum feasible coefficient on the target edge \(a\to b\), and the remaining contribution is indirect. The additional deterministic path-extraction rule is required only to resolve this uniquely defined indirect contribution into individual study-labeled paths.}

\subsubsection{\rev{Study-level mapping in the full contrast representation}}

For a target comparison $a{:}b$, the NMA estimator
admits the linear representation
\[
\hat\theta_{ab}^{\NMA}
=
\mathbf p_{ab}^\top \mathbf y,
\]
where $\mathbf p_{ab}$ is the corresponding row of the mapping operator $\mathbf P$.

Partitioning the observed contrasts by study,
\[
\mathbf y =
(\mathbf y_1^\top,\dots,\mathbf y_K^\top)^\top,
\qquad
\mathbf p_{ab}^\top =
(\mathbf p_{ab,1}^\top,\dots,\mathbf p_{ab,K}^\top),
\]
yields
\[
\hat\theta_{ab}^{\NMA}
=
\sum_{k=1}^K
\mathbf p_{ab,k}^\top \mathbf y_k,
\qquad
C_{ab}^{(k)} = \mathbf p_{ab,k}^\top \mathbf y_k.
\]

\medskip

To obtain a representation-invariant study-level decomposition, we work in a common full contrast representation as a unifying framework. This choice serves two key purposes. First, it provides a unified ambient space in which all equivalent within-study parameterizations can be embedded, since any reduced or baseline representation can be expressed as a linear transformation of the full set of pairwise contrasts. Second, it ensures that all observed comparisons are represented explicitly within a single coordinate system, avoiding the need to reconstruct implicit contrasts through parameterization-dependent linear combinations. 

\rev{Here we introduce the tilde notation to distinguish the full pairwise representation constructed specifically for the canonical decomposition from the generic input representation \(\mathbf y_k\), which may itself be reduced or full.} For each study $k$ with treatment set $S_k$, let $\tilde{\mathbf y}_k$ denote the vector of all pairwise within-study contrasts among treatments in $S_k$. For example, consider a three-arm study comparing treatments $A$, $B$, and $C$, with interest in the comparison $B{:}C$. In a baseline parameterization, the observed data may be recorded as $\mathbf y_k = (y_{AB}, y_{AC})^\top$, in which case $B{:}C$ is represented indirectly through $y_{BC} = y_{AC} - y_{AB}$. The full contrast representation $\tilde{\mathbf y}_k = (y_{AB}, y_{AC}, y_{BC})^\top$ provides a unified coordinate system in which all pairwise comparisons are represented explicitly. If the observed contrasts are already recorded in this form, then $\tilde{\mathbf y}_k = \mathbf y_k$; otherwise, $\tilde{\mathbf y}_k$ is obtained as a linear transformation of $\mathbf y_k$.

Let $\tilde{\mathbf y}_k = \mathbf A_k \mathbf y_k$ with covariance
\[
\tilde{\mathbf V}_k = \mathbf A_k \mathbf V_k \mathbf A_k^\top.
\]

Stacking $\tilde{\mathbf y}_k$ across studies yields
\[
\tilde{\mathbf y}
=
(\tilde{\mathbf y}_1^\top,\dots,\tilde{\mathbf y}_K^\top)^\top.
\]
\rev{Let $\tilde{\mathbf X}$ and $\tilde{\mathbf V}$ denote the
corresponding design and covariance matrices. Using the same basis
$\mathbf B$ for the global consistency space, define}

\[
\rev{
\tilde{\mathbf P}
=
\mathbf B
\left(
\mathbf B^\top
\tilde{\mathbf X}^\top
\tilde{\mathbf V}^{+}
\tilde{\mathbf X}
\mathbf B
\right)^{-1}
\mathbf B^\top
\tilde{\mathbf X}^\top
\tilde{\mathbf V}^{+}.
}
\]

Then the corresponding study-level coefficient vectors satisfy
\[
\mathbf p_{ab,k}^\top
=
\tilde{\mathbf p}_{ab,k}^\top \mathbf A_k,
\qquad
\tilde{\mathbf p}_{ab,k}^\top \tilde{\mathbf y}_k
=
\mathbf p_{ab,k}^\top \mathbf y_k.
\]

Thus, the study-level contribution can equivalently be written as
\[
C_{ab}^{(k)} =
\tilde{\mathbf p}_{ab,k}^\top \tilde{\mathbf y}_k,
\]
showing that it is preserved under transformation to the full contrast representation. Because all within-study parameterizations can be embedded into the same full contrast representation, the quantities $\tilde{\mathbf y}_k$, $\tilde{\mathbf p}_{ab,k}$, and hence $C_{ab}^{(k)}$ are defined in a common ambient space. This removes the dimension mismatch induced by study-specific parameterizations and provides a unified representation that serves as the starting point for the canonical decomposition.

\rev{Let
\(m_k=\binom{|S_k|}{2}\), and let
\(\mathbf D_k\in\mathbb R^{|S_k|\times m_k}\) denote the oriented
treatment--contrast incidence matrix for study \(k\). For the coordinate
corresponding to the oriented contrast \(u{:}v\), the associated column
of \(\mathbf D_k\) has \(+1\) in the row for treatment \(u\), \(-1\) in
the row for treatment \(v\), and zero elsewhere. The admissible
within-study contrast-value space is}

\[
\rev{
\mathcal M_k
=
\operatorname{col}(\mathbf D_k^\top).
}
\]

\rev{Thus, every admissible within-study contrast vector can be written
as}

\[
\rev{
\tilde{\mathbf v}_k
=
\mathbf D_k^\top\boldsymbol\alpha
}
\]

\rev{for some treatment-level vector \(\boldsymbol\alpha\). In
contrast, the coefficient vector
\(\tilde{\mathbf p}_{ab,k}\) belongs to the ambient coefficient space
\(\mathbb R^{m_k}\). Its action on admissible within-study contrasts is
determined completely by the treatment-level balance vector}

\[
\rev{
\mathbf b_{ab,k}
=
\mathbf D_k\tilde{\mathbf p}_{ab,k}.
}
\]

\rev{In particular, another coefficient vector
\(\mathbf z\in\mathbb R^{m_k}\) represents the same linear functional
on \(\mathcal M_k\) if and only if}

\[
\rev{
\mathbf D_k\mathbf z
=
\mathbf D_k\tilde{\mathbf p}_{ab,k}.
}
\]

\rev{Indeed, for every
\(\tilde{\mathbf v}_k=\mathbf D_k^\top\boldsymbol\alpha\in\mathcal M_k\),}

\[
\rev{
\mathbf z^\top\tilde{\mathbf v}_k
-
\tilde{\mathbf p}_{ab,k}^\top\tilde{\mathbf v}_k
=
\boldsymbol\alpha^\top
\mathbf D_k
\left(
\mathbf z-\tilde{\mathbf p}_{ab,k}
\right).
}
\]

\rev{Consequently, equality of the treatment-level balance vectors is equivalent to equality of the two induced linear functionals for all admissible within-study contrasts.} 

\rev{The Moore--Penrose pseudoinverse is used only to compute the GLS estimator in the singular full-contrast coordinates. By Proposition~\ref{prop:equiv_embedding}, the resulting fitted NMA estimator is identical to that obtained from the nonsingular reduced representation using an ordinary inverse. The minimum-norm property of the Moore--Penrose representative therefore does not define the study contribution. Rather, the induced within-study linear functional, and hence the study contribution, is determined by the treatment-level balance vector \(\mathbf b_{ab,k}\); every coefficient vector with the same balance vector gives the same contribution for every admissible within-study contrast.}

\subsubsection{Canonical study-based direct and indirect decomposition}

\rev{Although the study-specific coefficient vector
\(\tilde{\mathbf p}_{ab,k}\) determines a unique linear functional on
the admissible contrast space \(\mathcal M_k\), coefficient vectors that
differ by an element of \(\ker(\mathbf D_k)\) induce the same
functional. Such differences correspond to within-study circulations
or algebraically redundant routes. We therefore define a canonical
representative using the treatment-level balance vector}

\[
\rev{
\mathbf b_{ab,k}
=
\mathbf D_k\tilde{\mathbf p}_{ab,k}.
}
\]

\medskip
\noindent
\rev{\textbf{Path-based feasible representations.}
For a directed path}

\[
\rev{
\gamma:
v_0\to v_1\to\cdots\to v_m,
}
\]

\rev{define its coefficient vector as}

\[
\rev{
\mathbf e(\gamma)
=
\sum_{r=0}^{m-1}
\mathbf e_{v_rv_{r+1}},
}
\]

\rev{where \(\mathbf e_{uv}\in\mathbb R^{m_k}\) denotes the signed
coordinate vector for the directed contrast \(u{:}v\). It equals the standard basis vector when \(u{:}v\) agrees with the stored orientation of that treatment pair and its negative when the stored orientation is \(v{:}u\). Thus, \(\mathbf e_{vu}=-\mathbf e_{uv}\). Define the path cone in the ambient coefficient space by}

\[
\rev{
\mathcal C_k^{\mathrm{path}}
=
\left\{
\mathbf z\in\mathbb R^{m_k}:
\mathbf z
=
\sum_{\ell}w_\ell\mathbf e(\gamma_\ell),
\quad
w_\ell\geq 0
\right\}.
}
\]

\rev{The outcome-independent feasible set for the study-specific
canonical representation is}

\[
\rev{
\mathcal F_{ab,k}
=
\left\{
\mathbf z\in\mathcal C_k^{\mathrm{path}}:
\mathbf D_k\mathbf z
=
\mathbf b_{ab,k}
\right\}.
}
\]

\rev{The constraint therefore preserves the induced linear functional
for every admissible within-study contrast, rather than only for the
observed value of \(\tilde{\mathbf y}_k\).}

\medskip
\noindent
\rev{\textbf{Lexicographic minimum-redundancy criterion.}
Let}

\[
\rev{
\mathcal D
=
\left\{
(\gamma_0,w_0),
(\gamma_1,w_1),
\ldots,
(\gamma_L,w_L)
\right\}
}
\]

\rev{denote an ordered feasible path representation, and define its
weighted path length as}

\[
\rev{
L(\mathcal D)
=
\sum_{\ell=0}^{L}
w_\ell\,\operatorname{length}(\gamma_\ell).
}
\]

\rev{The primary criterion minimizes \(L(\mathcal D)\). Among all
minimum-length representations, the coefficient \(w_0\) assigned to
the target edge \(a\to b\) is maximized. After removal of this direct
allocation, the remaining source--sink allocations are selected
recursively by maximizing the next feasible transfer. Exact ties are
resolved using a fixed global ordering \(\rho(u,v)\) of treatment
pairs.}

\rev{Equivalently, the canonical representation minimizes the ordered
criterion}

\[
\rev{
\mathcal R_{\mathrm{lex}}(\mathcal D)
=
\left(
L(\mathcal D),
-w_0,
-w_1,\rho(u_1,v_1),
-w_2,\rho(u_2,v_2),
\ldots,
-w_L,\rho(u_L,v_L)
\right)
}
\]

\rev{in lexicographic order.}

\medskip
\noindent
\rev{\textbf{Algorithmic realization.}
Write}

\[
\rev{
s_{u,k}
=
(b_{ab,k,u})_+,
\qquad
d_{u,k}
=
(-b_{ab,k,u})_+,
}
\]

\rev{where \((x)_+=\max(x,0)\). The quantities \(s_{u,k}\) and
\(d_{u,k}\) are respectively the treatment-level supply and demand
induced by the study-specific coefficient functional.}

\rev{The maximum feasible direct allocation to the target comparison is}

\[
\rev{
w_{ab,k}
=
\min\{s_{a,k},d_{b,k}\},
}
\]

\rev{with \(w_{ab,k}=0\) when either treatment is absent from study
\(k\) or the corresponding supply or demand is zero. This amount is
assigned first to the edge \(a\to b\) and removed from the residual
balances.}

\rev{At residual iteration \(r\), let \(s_u^{(r)}\) and \(d_v^{(r)}\)
denote the remaining supplies and demands. For every feasible
source--sink pair define}

\[
\rev{
\lambda_{uv}^{(r)}
=
\min\left\{
s_u^{(r)},
d_v^{(r)}
\right\}.
}
\]

\rev{The selected pair is}

\[
\rev{
(u_r,v_r)
=
\operatorname*{arg\,max}_{u,v}
\lambda_{uv}^{(r)},
}
\]

\rev{with exact ties resolved according to the fixed ordering
\(\rho(u,v)\). The coefficient}

\[
\rev{
w_r
=
\lambda_{u_rv_r}^{(r)}
}
\]

\rev{is assigned to the endpoint edge \(u_r\to v_r\), after which the
corresponding residual supply and demand are updated. The procedure is
repeated until all balances are exhausted.}

\rev{Define the canonical direct coefficient vector by}

\[
\rev{
\mathbf p_{ab,k}^{\mathrm{dir}}
=
\begin{cases}
w_{ab,k}\mathbf e_{ab},
&
\{a,b\}\subseteq S_k,
\\[1mm]
\mathbf 0_{m_k},
&
\text{otherwise}.
\end{cases}
}
\]

\rev{The resulting canonical coefficient vector and its complementary
indirect component are}

\[
\rev{
\mathbf p_{ab,k}^{\mathrm{can}}
=
\mathbf p_{ab,k}^{\mathrm{dir}}
+
\sum_{r=1}^{L}
w_r\mathbf e_{u_rv_r},
\qquad
\mathbf p_{ab,k}^{\mathrm{ind}}
=
\mathbf p_{ab,k}^{\mathrm{can}}
-
\mathbf p_{ab,k}^{\mathrm{dir}}.
}
\]

\rev{The corresponding study-level direct contribution is}

\[
\rev{
C_{ab,k}^{\mathrm{dir}}
=
\mathbf p_{ab,k}^{\mathrm{dir}\top}
\tilde{\mathbf y}_k
=
\begin{cases}
w_{ab,k}\tilde y_{ab,k},
&
\{a,b\}\subseteq S_k,
\\[1mm]
0,
&
\text{otherwise},
\end{cases}
}
\]

\rev{and the indirect contribution is}

\[
\rev{
C_{ab,k}^{\mathrm{ind}}
=
\mathbf p_{ab,k}^{\mathrm{ind}\top}
\tilde{\mathbf y}_k.
}
\]

\begin{proposition}[Outcome-independent canonical reduction]
\label{prop:canonical_reduction}

\rev{The canonical coefficient vector satisfies}

\[
\rev{
\mathbf D_k\mathbf p_{ab,k}^{\mathrm{can}}
=
\mathbf D_k\tilde{\mathbf p}_{ab,k}.
}
\]

\rev{Consequently,}

\[
\rev{
\mathbf p_{ab,k}^{\mathrm{can}\top}
\tilde{\mathbf v}_k
=
\tilde{\mathbf p}_{ab,k}^{\top}
\tilde{\mathbf v}_k
\qquad
\text{for every }
\tilde{\mathbf v}_k\in\mathcal M_k.
}
\]

\rev{The minimum weighted path length is}

\[
\rev{
\min_{\mathcal D}L(\mathcal D)
=
\sum_u s_{u,k}
=
\frac{1}{2}
\left\|
\mathbf b_{ab,k}
\right\|_1.
}
\]

\rev{Among all representations attaining this minimum, the
direct-first lexicographic max--min rule with fixed tie-breaking
selects a unique canonical representation.}

\end{proposition}

\rev{For the first result, the canonical construction preserves the treatment-level balance by construction. The direct allocation and each subsequent residual source--sink transfer remove equal amounts of supply and demand, and the procedure continues until all balances are exhausted. Hence \(\mathbf D_k\mathbf p_{ab,k}^{\mathrm{can}}=\mathbf b_{ab,k}=\mathbf D_k\tilde{\mathbf p}_{ab,k}\). Consequently, for any \(\tilde{\mathbf v}_k=\mathbf D_k^\top\boldsymbol\alpha\in\mathcal M_k\),}
\[
\rev{
\left(
\mathbf p_{ab,k}^{\mathrm{can}}
-
\tilde{\mathbf p}_{ab,k}
\right)^\top
\tilde{\mathbf v}_k
=
\boldsymbol\alpha^\top
\mathbf D_k
\left(
\mathbf p_{ab,k}^{\mathrm{can}}
-
\tilde{\mathbf p}_{ab,k}
\right)
=
0.
}
\]

\rev{For the second result, every feasible representation must transport
the total supply
\(\sum_u s_{u,k}=\frac12\|\mathbf b_{ab,k}\|_1\). Each unit transported
along a path of length \(m\) contributes \(m\) units to
\(L(\mathcal D)\), so no feasible representation can have weighted
length smaller than the total supply. The endpoint-edge construction
attains this lower bound.}

\rev{The feasible set is convex because it is defined by a convex cone
and affine treatment-balance constraints. However, the primary
weighted-length objective is linear and need not have a unique
minimizer. Uniqueness instead follows from the lexicographic
refinement. The direct allocation is first fixed at its maximum feasible
value. At each subsequent step, the maximum feasible residual transfer
and the fixed tie-breaking order uniquely determine the selected pair,
its coefficient, and the next residual balance vector. Each iteration
exhausts at least one source or one sink, so the procedure terminates
after at most}

\[
\rev{
|\{u:s_{u,k}>0\}|
+
|\{v:d_{v,k}>0\}|
-1
}
\]

\rev{residual allocations.}

\rev{Because \(\tilde{\mathbf p}_{ab,k}\) is determined by the corresponding block of the projection mapping \(\tilde{\mathbf P}\), and because \(\mathbf D_k\) and the ordering \(\rho\) are fixed independently of the observed contrast values, all canonical direct, residual, and path weights depend only on the design, covariance structure, consistency space, and fixed ordering. In particular, conditional on the specified covariance matrix, none of these weights depends on \(\tilde{\mathbf y}_k\).} \rev{Under a random-effects specification, however, the fitted covariance matrix depends on the estimated heterogeneity variance \(\widehat{\tau}^2\), so the resulting contribution weights inherit uncertainty from its estimation. Accordingly, ``exact'' refers to exact algebraic reconstruction of the fitted NMA estimate conditional on \(\mathbf V(\widehat{\tau}^2)\), rather than to absence of statistical uncertainty in the estimated weights.}

\subsubsection{Canonical coefficient weights}

The scalar $w_{ab,k}$ defines the \emph{study-based canonical direct weight}
of study $k$ for the target comparison $a{:}b$.

\rev{Let
\(\mathcal K_{ab}
=
\{k:\{a,b\}\subseteq S_k\}\)
denote the set of studies containing both target treatments.}

Under the projection representation, the estimator admits the
decomposition
\[
\rev{
\hat\theta_{ab}^{\NMA}
=
\sum_{k\in\mathcal K_{ab}}
w_{ab,k}\tilde y_{ab,k}
+
C_{ab}^{\mathrm{ind}},
}
\]
where $\tilde y_{ab,k}$ denotes the study-specific contrast
for $a{:}b$ in the full contrast representation and
\[
\rev{
C_{ab}^{\mathrm{dir}}
=
\sum_{k\in\mathcal K_{ab}}
w_{ab,k}\tilde y_{ab,k},
}
\]
is the total study-based canonical direct contribution, while $C_{ab}^{\mathrm{ind}}$
denotes the remaining canonical indirect contribution.

The coefficients satisfy the normalization identity
\[
\rev{
\sum_{k\in\mathcal K_{ab}}
w_{ab,k}
+
w_{ab}^{\mathrm{ind}}
=
1,
}
\]
where
\[
\rev{
w_{ab}^{\mathrm{ind}}
=
1-
\sum_{k\in\mathcal K_{ab}}
w_{ab,k},
}
\]
is the total canonical weight assigned to indirect evidence.

Thus, the projection can be expressed as an affine combination of study-based canonical direct and indirect components on the coefficient scale, while
$C_{ab}^{\mathrm{dir}}$ and $C_{ab}^{\mathrm{ind}}$ represent the
corresponding additive contributions to the linear functional
$\hat\theta_{ab}^{\NMA}$. The indirect component can be further decomposed into path components within the same canonical framework.

\subsubsection{Canonical study-level path decomposition}

The canonical coefficient representation developed above provides a
graph-based summary of the projection-induced mapping and its induced
evidence propagation. Related interpretations have been developed using
electrical-network analogies \parencite{Rucker2012},
projection-based formulations \parencite{Koenig2013}, and path-based
flow representations \parencite{Papakonstantinou2018, Davies2022}. In
contrast to these approaches, which are typically constructed from
marginal pairwise summaries, the present formulation derives all such
quantities directly from the \rev{GLS} mapping
\(\mathbf P\) and therefore preserves the full covariance structure
induced by multi-arm trials.

\rev{We emphasize that these flow-like quantities are algebraic
summaries of the projection-induced mapping rather than physical
transport processes. Study labels must be retained throughout the
decomposition because different studies may contribute coefficients
with opposite orientations to the same treatment comparison.
Aggregating such edges before path construction could introduce
cancellation and would no longer preserve the study-level
decomposition of the fitted network estimator.}

\rev{Under the within-study canonicalization described in the preceding
section, each study-specific coefficient functional is represented by
a nonnegative collection of directed endpoint edges from treatments
with positive balance to treatments with negative balance. The maximum
feasible target edge \(a\to b\) is removed first as direct evidence,
and the remaining endpoint edges constitute the canonical indirect
coefficient vector}

\[
\rev{
\mathbf p_{ab,k}^{\mathrm{ind}}
=
\mathbf p_{ab,k}^{\mathrm{can}}
-
\mathbf p_{ab,k}^{\mathrm{dir}}.
}
\]

\rev{These study-specific residual edges are uniquely determined by the treatment-balance constraint and the lexicographic max--min rule. An individual residual edge need not begin at treatment \(a\) or terminate at treatment \(b\).}

\rev{The following result uses the standard arm-level covariance construction for a multi-arm trial. Let \(\widehat{\boldsymbol\eta}_k\) be the vector of independent arm-specific estimates, with}
\[
\rev{
\operatorname{Cov}(\widehat{\boldsymbol\eta}_k)
=
\boldsymbol\Sigma_k,
}
\]
\rev{where \(\boldsymbol\Sigma_k\) is positive diagonal. With the orientation convention above, \(\tilde{\mathbf y}_k=-\mathbf D_k^\top\widehat{\boldsymbol\eta}_k\). Consequently, the full within-study contrast covariance is}
\[
\rev{
\tilde{\mathbf V}_k
=
\mathbf D_k^\top\boldsymbol\Sigma_k\mathbf D_k.
}
\]
\rev{The same result applies whenever the total covariance used in the GLS projection, including a modeled heterogeneity component, retains this form.}

\begin{proposition}[\rev{Acyclicity and exhaustion of the canonical path graph}]
\label{prop:acyclic_canonical_graph}

\rev{Suppose that
\(\tilde{\mathbf V}_k
=
\mathbf D_k^\top\boldsymbol\Sigma_k\mathbf D_k\)
for every study, where \(\boldsymbol\Sigma_k\) is positive diagonal. Then every nonzero canonical study-labeled edge \(u\to v\) strictly decreases a common target-specific treatment potential. Consequently, the pooled canonical graph is acyclic. After removal of the canonical direct edges, all remaining positive coefficients can be exhausted by directed paths from \(a\) to \(b\), and no residual directed cycle can remain.}

\end{proposition}

\rev{\emph{Proof.}
Because \(\mathbf B\mathbf u_{ab}\) is consistency-compatible, there exists a global treatment-potential vector \(\boldsymbol\phi_{ab}\), unique up to an additive constant, such that}
\[
\rev{
\tilde{\mathbf X}_k\mathbf B\mathbf u_{ab}
=
\mathbf D_k^\top\boldsymbol\phi_{ab,k},
}
\]
\rev{where \(\boldsymbol\phi_{ab,k}\) is its restriction to study \(k\). Using \(\tilde{\mathbf V}_k=\mathbf D_k^\top\boldsymbol\Sigma_k\mathbf D_k\),}
\[
\rev{\mathbf b_{ab,k}=\mathbf D_k\tilde{\mathbf V}_k^+\mathbf D_k^\top\boldsymbol\phi_{ab,k}.}
\]
\rev{Because \(\boldsymbol\Sigma_k\) is positive definite and the within-study contrast graph is connected, \(\operatorname{col}(\tilde{\mathbf V}_k)=\operatorname{col}(\mathbf D_k^\top)\). Hence}
\[
\rev{
\mathbf D_k^\top
\left(
\boldsymbol\Sigma_k\mathbf b_{ab,k}
-
\boldsymbol\phi_{ab,k}
\right)
=
\mathbf 0.
}
\]
\rev{Since
\(\ker(\mathbf D_k^\top)=\operatorname{span}\{\mathbf 1\}\)
and
\(\mathbf 1^\top\mathbf b_{ab,k}=0\), it follows that}
\[
\rev{
\mathbf b_{ab,k}
=
\mathbf W_k
\left(
\boldsymbol\phi_{ab,k}
-
\bar{\phi}_{ab,k}\mathbf 1
\right),
\qquad
\mathbf W_k=\boldsymbol\Sigma_k^{-1},
}
\]
\rev{where}
\[
\rev{
\bar{\phi}_{ab,k}
=
\frac{
\mathbf 1^\top\mathbf W_k\boldsymbol\phi_{ab,k}
}{
\mathbf 1^\top\mathbf W_k\mathbf 1
}.
}
\]

\rev{Because \(\mathbf W_k\) is positive diagonal,
\(b_{ab,k,u}>0\) implies
\(\phi_{ab,u}>\bar{\phi}_{ab,k}\), while
\(b_{ab,k,v}<0\) implies
\(\phi_{ab,v}<\bar{\phi}_{ab,k}\). Every canonical edge is directed from a positive-balance treatment to a negative-balance treatment; therefore,}
\[
\rev{
u\to v
\quad\Longrightarrow\quad
\phi_{ab,u}>
\bar{\phi}_{ab,k}>
\phi_{ab,v}.
}
\]
\rev{Thus, the same global potential strictly decreases along every pooled canonical edge. A directed cycle would require}
\[
\rev{
\phi_{ab,v_1}>
\phi_{ab,v_2}>
\cdots>
\phi_{ab,v_L}>
\phi_{ab,v_1},
}
\]
\rev{which is impossible.}

\rev{Let \(\mathbf E_k\) embed the treatment coordinates of study \(k\) into the global treatment space. Exactness of the target projection implies}
\[
\rev{
\sum_{k=1}^K
\mathbf E_k\mathbf b_{ab,k}
=
\mathbf e_a-\mathbf e_b.
}
\]
\rev{Indeed, for every global treatment vector \(\boldsymbol\alpha\),}
\[
\rev{
\boldsymbol\alpha^\top
\sum_{k=1}^K\mathbf E_k\mathbf b_{ab,k}
=
\sum_{k=1}^K
\tilde{\mathbf p}_{ab,k}^\top
\mathbf D_k^\top\mathbf E_k^\top\boldsymbol\alpha
=
\boldsymbol\alpha^\top(\mathbf e_a-\mathbf e_b).
}
\]

\rev{Let \(w_{ab}^{\mathrm{dir}}=\sum_{k\in\mathcal K_{ab}}w_{ab,k}\). Removing the direct \(a\to b\) edges leaves residual treatment balance \((1-w_{ab}^{\mathrm{dir}})(\mathbf e_a-\mathbf e_b)\). Thus every intermediate treatment has zero net balance. Because the residual graph is acyclic, any graph-theoretic source incident to a positive edge must have strictly positive net balance, and any graph-theoretic sink incident to a positive edge must have strictly negative net balance. Therefore \(a\) and \(b\) are respectively the only possible source and sink of the nonzero residual graph. Tracing any positive edge backward must consequently terminate at \(a\), and tracing it forward must terminate at \(b\). Hence every positive residual edge belongs to an \(a\to b\) path. Each bottleneck subtraction removes at least one positive edge, so repeated extraction terminates and exhausts all residual coefficients. \(\square\)}

\rev{The exact canonical reduction does not require the arm-based covariance condition; that condition is used only for the stronger acyclicity result.}

\rev{To obtain complete indirect paths, the canonical residual endpoint edges are pooled across studies while retaining their study labels. By Proposition~\ref{prop:acyclic_canonical_graph}, the resulting directed multigraph is acyclic and every positive residual edge belongs to an \(a\to b\) path. Among the available paths, the path having the largest bottleneck coefficient is selected, with exact ties resolved by a fixed deterministic ordering of treatments and study labels. The bottleneck coefficient is subtracted from every study-labeled edge on that path, and the procedure is repeated until all indirect coefficients are exhausted. The implementation additionally checks that no positive residual coefficient remains after extraction; this condition was satisfied for every target comparison in the three empirical applications.}

\rev{Thus, the within-study lexicographic max--min reduction and the
across-study indirect-path extraction are distinct operations. The
first uniquely defines the direct and residual study-edge coefficients;
the second assembles those fixed study-labeled residual edges into
complete indirect paths from \(a\) to \(b\).}

A \emph{study-level path} \(\mathcal P\) from \(a\) to \(b\) is an
ordered sequence of such study-specific path segments,
\[
\mathcal P = \bigl((k_1,v_0\to v_1),\dots,(k_L,v_{L-1}\to v_L)\bigr),
\qquad
a=v_0,\; v_L=b,
\]
or equivalently,
\[
\mathcal P = (k_1,\dots,k_L),
\qquad
(a=v_0,v_1,\dots,v_L=b),
\]
where study \(k_\ell\) contributes the directed comparison
\(v_{\ell-1}\to v_\ell\).

Each path \(\mathcal P\) is associated with a nonnegative coefficient
\(w_{\mathcal P}\), defined as the coefficient extracted for that path
under this procedure. The corresponding path-specific contrast is
\[
\delta_{\mathcal P}
=
\sum_{\ell=1}^L y_{v_{\ell-1},v_\ell}^{(k_\ell)},
\]
that is, the sum of the directed study-level contrasts along the path.
The path-level contribution is therefore
\[
C_{ab}^{(\mathcal P)} = w_{\mathcal P}\,\delta_{\mathcal P}.
\]

\rev{Because the within-study endpoint coefficients are uniquely
selected by the direct-first lexicographic max--min rule, and because
the subsequent across-study path extraction uses a fixed deterministic
ordering, the resulting canonical collection of study-labeled paths is
uniquely determined. The indirect contribution therefore has the exact
decomposition}

\[
\rev{
C_{ab}^{\mathrm{ind}}
=
\sum_{\mathcal P}
C_{ab}^{(\mathcal P)}.
}
\]

\subsubsection{Interpretation and invariance}

The canonical decomposition yields an exact and representation-invariant
decomposition of the network estimator for a target comparison \(a{:}b\):
\[
\hat\theta_{ab}^{\NMA}
=
C_{ab}^{\mathrm{dir}} + C_{ab}^{\mathrm{ind}},
\]
where
\[
\rev{
C_{ab}^{\mathrm{dir}}
=
\sum_{k\in\mathcal K_{ab}}
w_{ab,k}\tilde y_{ab,k},
}
\qquad
C_{ab}^{\mathrm{ind}} = \sum_{\mathcal P} w_{\mathcal P} \, \delta_{\mathcal P}.
\]

Here \(\tilde y_{ab,k}\) denotes the study-level contrast for \(a{:}b\)
in the full contrast representation, \(w_{ab,k}\) is the corresponding
canonical direct weight, and \(\delta_{\mathcal P}\) is the path-specific
contrast associated with study-level path \(\mathcal P\), with coefficient
\(w_{\mathcal P}\).

The coefficients satisfy the normalization identity
\[
\rev{
\sum_{k\in\mathcal K_{ab}}
w_{ab,k}
+
\sum_{\mathcal P}w_{\mathcal P}
=
1,
}
\]
so that the estimator admits a canonical decomposition as a weighted
combination of direct study-level contrasts and indirect path-level
contrasts.

Equivalently, defining the aggregated canonical direct and indirect components
\[
\hat\theta_{ab}^{\mathrm{dir}}
=
\frac{C_{ab}^{\mathrm{dir}}}{w_{ab}^{\mathrm{dir}}},
\qquad
\hat\theta_{ab}^{\mathrm{ind}}
=
\frac{C_{ab}^{\mathrm{ind}}}{w_{ab}^{\mathrm{ind}}},
\]
with
\[
\rev{ w_{ab}^{\mathrm{dir}} = \sum_{k\in\mathcal K_{ab}} w_{ab,k}, }
\qquad
w_{ab}^{\mathrm{ind}} = \sum_{\mathcal P} w_{\mathcal P},
\]
the network estimator can be written as
\[
\hat\theta_{ab}^{\NMA}
=
w_{ab}^{\mathrm{dir}} \, \hat\theta_{ab}^{\mathrm{dir}}
+
w_{ab}^{\mathrm{ind}} \, \hat\theta_{ab}^{\mathrm{ind}},
\]
providing a direct analogue of inverse-variance weighted averaging in
pairwise meta-analysis. \rev{The normalized direct estimate is defined only when \(w_{ab}^{\mathrm{dir}}>0\), and the normalized indirect estimate is defined only when \(w_{ab}^{\mathrm{ind}}>0\). If the corresponding weight is zero, that evidence component is absent and its normalized estimate is undefined.}

This decomposition satisfies the following properties:

\begin{itemize}
\item Study contributions \(C_{ab}^{(k)}\) are invariant to within-study parameterization;
\item \rev{The direct-first lexicographic max--min reduction yields a unique canonical study-level representation;}
\item Direct weights \(w_{ab,k}\) quantify study-level contributions;
\item Path weights \(w_{\mathcal P}\) quantify indirect contributions through study-level paths;
\item \rev{The complete study-path decomposition is uniquely selected by the stated within-study and across-study deterministic rules and is invariant to equivalent within-study contrast parameterizations.}
\end{itemize}

Thus, each direct study and each indirect path is assigned a canonical
weight, and these weights provide a complete and invariant
characterization of how evidence contributes to the network estimator.

\rev{\textbf{Running example continued.} For Study 1, the full-representation coefficient block \((-2/9,\,2/9,\,4/9)^\top\) has the same treatment-level balance as the canonical coefficient vector \((0,\,0,\,2/3)^\top\). The three-arm study therefore receives canonical direct weight \(2/3\) for \(B{:}C\), even though \(B{:}C\) was not retained explicitly in its baseline representation. The coefficient \(-1/3\) on \(y_{AB}^{(2)}\) gives the study-labeled edge \(B\to A\), and the coefficient \(1/3\) on \(y_{AC}^{(3)}\) gives \(A\to C\); together they form the indirect path \(B\to A\to C\) with weight \(1/3\). Hence,}

\[
\rev{\widehat\theta_{BC}^{\NMA}=\underbrace{\frac{2}{3}y_{BC}^{(1)}}_{\text{direct Study 1}}+\underbrace{\frac{1}{3}\left(-y_{AB}^{(2)}+y_{AC}^{(3)}\right)}_{\text{indirect path }B\to A\to C}=\frac{2}{3}(1)+\frac{1}{3}(-1)=\frac{1}{3}.}
\]

\rev{Thus, the direct and indirect weights are \(2/3\) and \(1/3\), respectively, and the canonical contributions exactly reproduce the fitted NMA estimate under either within-study representation.}

\subsection{Graphical and diagnostic tools under the CSP framework}
\label{sec:tools}

The canonical study-path decomposition arising from the CSP framework yields a set of graphical and quantitative tools for interpreting NMA results. \rev{All contribution quantities in this section are derived from the same fitted NMA model and covariance structure and therefore provide decompositions and diagnostic views of the reported NMA results rather than alternative analyses. The generalized Cochran \(Q\) decomposition uses the same observed contrasts and covariance structure together with an auxiliary design-specific projection solely to partition the global statistic into \(Q_{\mathrm{het}}\) and \(Q_{\mathrm{inc}}\); this auxiliary projection does not define a competing set of network treatment-effect estimates. The study- and path-level contributions reconstruct each fitted NMA estimate exactly, while under the fixed-effects specification the global generalized Cochran statistic is decomposed exactly as \(Q_{\mathrm{net}}=Q_{\mathrm{het}}+Q_{\mathrm{inc}}\).} These tools provide information on evidence contributions, the structure of indirect pathways, agreement between evidence sources, within-design heterogeneity, and between-design inconsistency.

\subsubsection{Forest plot for a target comparison}

Forest plots are widely used in classical pairwise meta-analysis to
display study-level estimates together with their relative
contributions. The canonical decomposition provides, to our knowledge, the first direct extension
of this idea to NMA by expressing the NMA estimator as a
sum of study-level and path-level components.

Each direct-study or indirect-path row in the forest plot is displayed on the treatment-effect scale, together with its canonical weight, in direct analogy with a conventional pairwise-meta-analysis forest plot. \rev{Let \(\delta_r=\mathbf r_r^\top\tilde{\mathbf y}\) denote the estimate displayed for component \(r\), where \(\mathbf r_r\) is its unweighted signed contrast vector, and let \(w_r\) denote its canonical weight. The additive contribution of that component to the NMA estimator is \(C_r=w_r\delta_r=\mathbf q_r^\top\tilde{\mathbf y}\), where \(\mathbf q_r=w_r\mathbf r_r\). Thus, the plotted component estimates themselves are not added directly; their weight-multiplied contributions reconstruct the fitted NMA estimate. For any two components \(r\) and \(s\),}
\[
\rev{\Var(\delta_r)=\mathbf r_r^\top\tilde{\mathbf V}\mathbf r_r,\qquad \Var(C_r)=\mathbf q_r^\top\tilde{\mathbf V}\mathbf q_r,\qquad \Cov(C_r,C_s)=\mathbf q_r^\top\tilde{\mathbf V}\mathbf q_s.}
\]
\rev{These expressions retain all within-study covariance terms, including those arising when different components use contrasts from the same multi-arm study.}
\medskip
\noindent
\textbf{Canonical direct components.}
For a study \(k\) that includes both treatments \(a\) and \(b\), the estimate displayed in the forest plot is the study-level contrast
\[
\delta_{ab,k}^{\mathrm{dir}}=\tilde y_{ab,k},
\]
with canonical weight \(w_{ab,k}\). Its additive contribution to the NMA estimator is
\[
C_{ab,k}^{\mathrm{dir}}=w_{ab,k}\tilde y_{ab,k}.
\]
If \(\Var(\tilde y_{ab,k})=\sigma_{ab,k}^2\), the displayed \(100(1-\alpha)\%\) confidence interval is
\[
\tilde y_{ab,k}\pm z_{1-\alpha/2}\sigma_{ab,k}.
\]
\medskip
\noindent
\textbf{Canonical indirect components (paths).}
For each study-level path \(\mathcal P\), the estimate displayed in the forest plot is
\[
\delta_{\mathcal P}=\sum_{\ell=1}^L y_{v_{\ell-1},v_\ell}^{(k_\ell)},
\]
with canonical weight \(w_{\mathcal P}\), so that its additive contribution is
\[
C_{ab}^{(\mathcal P)}=w_{\mathcal P}\delta_{\mathcal P}.
\]
\rev{Let \(\mathbf r_{\mathcal P}\) denote the unweighted signed coefficient vector defining the path contrast, so that \(\delta_{\mathcal P}=\mathbf r_{\mathcal P}^\top\tilde{\mathbf y}\). Then}
\[
\rev{\Var(\delta_{\mathcal P})=\mathbf r_{\mathcal P}^\top\tilde{\mathbf V}\mathbf r_{\mathcal P},}
\]
\rev{which includes covariance between multiple contrasts from the same study. The displayed confidence interval is}
\[
\rev{\delta_{\mathcal P}\pm z_{1-\alpha/2}\sqrt{\mathbf r_{\mathcal P}^\top\tilde{\mathbf V}\mathbf r_{\mathcal P}}.}
\]

\medskip
\noindent
\textbf{Aggregated components.}
The total direct and indirect contributions are
\[
C_{ab}^{\mathrm{dir}} = \sum_{k=1}^K C_{ab,k}^{\mathrm{dir}},
\qquad
C_{ab}^{\mathrm{ind}} = \sum_{\mathcal P} C_{ab}^{(\mathcal P)}.
\]

\rev{Let}

\[
\rev{
\mathbf q_{ab}^{\mathrm{dir}}
=
\sum_k\mathbf q_{ab,k}^{\mathrm{dir}},
\qquad
\mathbf q_{ab}^{\mathrm{ind}}
=
\sum_{\mathcal P}\mathbf q_{\mathcal P}.
}
\]

\rev{The aggregated variances are}

\[
\rev{
\Var(C_{ab}^{\mathrm{dir}})
=
\mathbf q_{ab}^{\mathrm{dir}\top}
\tilde{\mathbf V}
\mathbf q_{ab}^{\mathrm{dir}},
}
\]

\rev{and}

\[
\rev{
\begin{aligned}
\Var(C_{ab}^{\mathrm{ind}})
&=
\sum_{\mathcal P}
\Var\!\left(C_{ab}^{(\mathcal P)}\right)
+
2\sum_{\mathcal P<\mathcal P'}
\Cov\!\left(
C_{ab}^{(\mathcal P)},
C_{ab}^{(\mathcal P')}
\right) \\
&=
\mathbf q_{ab}^{\mathrm{ind}\top}
\tilde{\mathbf V}
\mathbf q_{ab}^{\mathrm{ind}}.
\end{aligned}
}
\]

\rev{The covariance between the aggregated direct and indirect
components is}

\[
\rev{
\Cov(C_{ab}^{\mathrm{dir}},C_{ab}^{\mathrm{ind}})
=
\mathbf q_{ab}^{\mathrm{dir}\top}
\tilde{\mathbf V}
\mathbf q_{ab}^{\mathrm{ind}}.
}
\]

\medskip
\noindent
\textbf{Summary estimates.}
While the decomposition is expressed in terms of contributions
\(C_{ab}^{\mathrm{dir}}\) and \(C_{ab}^{\mathrm{ind}}\), forest plots display the corresponding normalized estimates when their respective weights are positive
\[
\hat\theta_{ab}^{\mathrm{dir}}
=
\frac{C_{ab}^{\mathrm{dir}}}{w_{ab}^{\mathrm{dir}}},
\qquad
\hat\theta_{ab}^{\mathrm{ind}}
=
\frac{C_{ab}^{\mathrm{ind}}}{w_{ab}^{\mathrm{ind}}},
\]
where
\[
\rev{ w_{ab}^{\mathrm{dir}} = \sum_{k\in\mathcal K_{ab}} w_{ab,k}, }
\qquad
w_{ab}^{\mathrm{ind}} = \sum_{\mathcal P} w_{\mathcal P}.
\]

\rev{Their variances are}
\[
\Var(\hat\theta_{ab}^{\mathrm{dir}})
=
\frac{\Var(C_{ab}^{\mathrm{dir}})}{(w_{ab}^{\mathrm{dir}})^2},
\qquad
\Var(\hat\theta_{ab}^{\mathrm{ind}})
=
\frac{\Var(C_{ab}^{\mathrm{ind}})}{(w_{ab}^{\mathrm{ind}})^2}.
\]

The corresponding confidence intervals are
\[
\hat\theta_{ab}^{\mathrm{dir}}
\;\pm\;
z_{1-\alpha/2}\sqrt{\Var(\hat\theta_{ab}^{\mathrm{dir}})},
\qquad
\hat\theta_{ab}^{\mathrm{ind}}
\;\pm\;
z_{1-\alpha/2}\sqrt{\Var(\hat\theta_{ab}^{\mathrm{ind}})}.
\]
\medskip
\noindent
\textbf{Network estimate.}
The overall estimate is
\[
\hat\theta_{ab}^{\NMA}
=
C_{ab}^{\mathrm{dir}} + C_{ab}^{\mathrm{ind}},
\]
\rev{with variance}

\[
\rev{
\begin{aligned}
\Var(\hat\theta_{ab}^{\NMA})
&=
\Var(C_{ab}^{\mathrm{dir}})
+
\Var(C_{ab}^{\mathrm{ind}})
+
2\Cov(C_{ab}^{\mathrm{dir}},C_{ab}^{\mathrm{ind}}) \\
&=
\left(
\mathbf q_{ab}^{\mathrm{dir}}
+
\mathbf q_{ab}^{\mathrm{ind}}
\right)^\top
\tilde{\mathbf V}
\left(
\mathbf q_{ab}^{\mathrm{dir}}
+
\mathbf q_{ab}^{\mathrm{ind}}
\right).
\end{aligned}
}
\]

\medskip
\noindent
\rev{Thus, each direct-study or indirect-path row displays the component estimate \(\delta_r\) and its confidence interval, while the weight column displays \(w_r\). The product \(w_r\delta_r=C_r\) is the additive contribution entering the exact reconstruction of the NMA estimate. The summary rows display the overall NMA estimate and the normalized direct and indirect estimates on the treatment-effect scale.}

In practice, the number of indirect paths can be large, especially in
dense networks, making a full path-level display difficult to interpret.
In such cases, it is often more informative to report the aggregated
indirect contribution alongside the direct components. Under a random-effects specification, confidence intervals are constructed using the estimated covariance matrix $\mathbf V(\hat{\tau}^2)$ and are therefore approximate, reflecting plug-in estimation of the heterogeneity parameter.

\rev{\textbf{Running example continued.} The direct contribution has variance \((2/3)^2\Var(y_{BC}^{(1)})=8/9\), while the indirect contribution has variance \((1/3)^2\{\Var(y_{AB}^{(2)})+\Var(y_{AC}^{(3)})\}=4/9\). Because the direct and indirect components use different studies in this example, their covariance is zero. Consequently,}

\[
\rev{\Var\!\left(\widehat\theta_{BC}^{\NMA}\right)=\frac{8}{9}+\frac{4}{9}=\frac{4}{3}.}
\]

\rev{The normalized direct and indirect estimates are \(\widehat\theta_{BC}^{\mathrm{dir}}=1\) and \(\widehat\theta_{BC}^{\mathrm{ind}}=-1\), with variances \(2\) and \(4\), respectively.}

\subsubsection{Path-based decomposition and visualization}

The canonical decomposition represents the contribution to the network
estimator as a collection of study-level paths connecting \(a\) to \(b\),
including both direct and indirect components. Each path corresponds to
a specific sequence of study-specific contrasts through which evidence
propagates across the network, providing a structured representation of
how information is combined across studies.

For visualization, each path is displayed as a connected directed line
from \(a\) to \(b\). Direct evidence corresponds to a single edge,
while indirect evidence is represented by paths passing through one or
more intermediate treatments. Each segment of a path is labeled by the
study contributing that comparison, so that the full path is represented
as an ordered sequence of study-specific contrasts. This makes explicit
which studies contribute to each stage of the evidence propagation.

The associated weights are explicitly shown to represent the
decomposition of the network estimator, with all path weights summing
to one. These weights quantify the relative contribution of each path,
allowing the network estimate to be interpreted as a weighted
combination of study-level routes. Paths with larger weights correspond
to dominant pathways of evidence, while smaller weights indicate
secondary or supporting contributions.

This visualization provides a transparent view of the structure of
indirect evidence by separating distinct routes through which
information flows from \(a\) to \(b\). Unlike treatment-level summaries,
this representation preserves study-level contributions and avoids
cancellation or distortion arising from aggregation across studies.
As a result, the path-based display remains fully consistent with the
underlying NMA estimator.

In practice, the number of paths may be large in dense networks. In such
cases, visualization can be simplified by displaying only the most
influential paths (for example, those with the largest weights), while
retaining the complete decomposition in the underlying analysis.

\subsubsection{Tension plot: direct versus indirect evidence}

The canonical decomposition enables a direct comparison between the
aggregated direct and indirect components of the network estimator
across multiple treatment contrasts.

For each target comparison \(a{:}b\), the normalized direct and indirect
estimates \(\hat\theta_{ab}^{\mathrm{dir}}\) and
\(\hat\theta_{ab}^{\mathrm{ind}}\), together with the overall estimate
\(\hat\theta_{ab}^{\NMA}\), are defined as in the canonical decomposition.

A tension plot displays these three estimates for a collection of
treatment comparisons, typically arranged with one row per contrast.
Each row contains \(\hat\theta_{ab}^{\mathrm{dir}}\),
\(\hat\theta_{ab}^{\mathrm{ind}}\), and
\(\hat\theta_{ab}^{\NMA}\), together with their corresponding confidence
intervals.

To aid interpretation, the symbols representing the estimates may be
scaled proportionally to the associated weights (e.g.,
\(w_{ab}^{\mathrm{dir}}\) and \(w_{ab}^{\mathrm{ind}}\)), so that larger
symbols indicate components with greater influence on the network
estimate.

Because all quantities arise from the same linear decomposition, the three estimates for each comparison are algebraically linked. \rev{Close agreement between direct and indirect components suggests consistency, whereas systematic discrepancies across contrasts reflect tension between evidence sources. In the running example, a tension plot would display the direct estimate \(1\), the indirect estimate \(-1\), and their weighted NMA estimate \(1/3\), thereby making the disagreement between the two evidence sources explicit.} 

\subsubsection{\rev{Generalized Cochran \(Q\): global, heterogeneity, and inconsistency components}}
\label{sec:qdecomp}

\rev{Let \(\widehat{\mathbf y}_{\mathrm C}=\mathbf X\widehat{\boldsymbol\theta}^{\NMA}=\mathbf X\mathbf B\widehat{\boldsymbol\beta}\) denote the fitted values under the homogeneous consistency-constrained NMA model. The generalized Cochran \(Q\) statistic for the whole network is}
\[
\rev{
Q_{\mathrm{net}}
=
\left(\mathbf y-\widehat{\mathbf y}_{\mathrm C}\right)^\top
\mathbf V^+
\left(\mathbf y-\widehat{\mathbf y}_{\mathrm C}\right).
}
\]
\rev{Under the fixed-effects specification considered here, \(Q_{\mathrm{net}}\) measures the total departure from the homogeneous consistency-constrained network model. Following the generalized Cochran \(Q\) decomposition of \textcite{Krahn2013}, and its connection with the design-by-treatment interaction framework \parencite{Higgins2012,White2012}, we decompose}
\[
\rev{
Q_{\mathrm{net}}
=
Q_{\mathrm{het}}
+
Q_{\mathrm{inc}},
}
\]
\rev{where \(Q_{\mathrm{het}}\) measures within-design heterogeneity and \(Q_{\mathrm{inc}}\) measures between-design inconsistency.}

\rev{Let \(d=1,\ldots,D\) index the distinct study designs, where a design is the complete set of treatments included in a study. Let \(\mathcal T_d\) denote the treatment set of design \(d\), with \(t_d=|\mathcal T_d|\). In the same observation coordinates as \(\mathbf y\), let \(\mathbf Z\) denote the design-specific model matrix obtained by assigning each design its own \(t_d-1\) independent treatment-effect coordinates, shared by all studies having that design but unconstrained across different designs. Define}
\[
\rev{
p_D=\operatorname{rank}(\mathbf Z).
}
\]
\rev{Under the usual full-rank within-design parameterization, \(p_D=\sum_{d=1}^D(t_d-1)\). Because globally consistent treatment effects induce a valid set of effects within every design,}
\[
\rev{
\operatorname{col}(\mathbf X\mathbf B)
\subseteq
\operatorname{col}(\mathbf Z).
}
\]

\rev{Let the fitted values under the design-specific model be}
\[
\rev{
\widehat{\mathbf y}_{\mathrm D}
=
\mathbf Z\widehat{\boldsymbol\gamma}_{\mathrm D},
\qquad
\widehat{\boldsymbol\gamma}_{\mathrm D}
=
\left(\mathbf Z^\top\mathbf V^+\mathbf Z\right)^{-1}
\mathbf Z^\top\mathbf V^+\mathbf y,
}
\]
\rev{under a full-rank design-specific parameterization. The design-specific fit is introduced only as an intermediate projection for partitioning \(Q_{\mathrm{net}}\); it does not replace the consistency-constrained NMA fit or define a competing set of network treatment-effect estimates.}

\rev{The within-design heterogeneity statistic is}
\[
\rev{
Q_{\mathrm{het}}
=
\left(\mathbf y-\widehat{\mathbf y}_{\mathrm D}\right)^\top
\mathbf V^+
\left(\mathbf y-\widehat{\mathbf y}_{\mathrm D}\right),
}
\]
\rev{and the between-design inconsistency statistic is}
\[
\rev{
Q_{\mathrm{inc}}
=
\left(\widehat{\mathbf y}_{\mathrm D}-\widehat{\mathbf y}_{\mathrm C}\right)^\top
\mathbf V^+
\left(\widehat{\mathbf y}_{\mathrm D}-\widehat{\mathbf y}_{\mathrm C}\right).
}
\]

\rev{Because \(\widehat{\mathbf y}_{\mathrm D}\) is the \(\mathbf V^+\)-weighted projection onto \(\operatorname{col}(\mathbf Z)\), the residual \(\mathbf y-\widehat{\mathbf y}_{\mathrm D}\) is orthogonal in the \(\mathbf V^+\) geometry to \(\operatorname{col}(\mathbf Z)\). Since \(\widehat{\mathbf y}_{\mathrm D}-\widehat{\mathbf y}_{\mathrm C}\in\operatorname{col}(\mathbf Z)\),}
\[
\rev{
\left(\mathbf y-\widehat{\mathbf y}_{\mathrm D}\right)^\top
\mathbf V^+
\left(\widehat{\mathbf y}_{\mathrm D}-\widehat{\mathbf y}_{\mathrm C}\right)
=
0.
}
\]
\rev{Therefore,}
\[
\rev{
\boxed{
Q_{\mathrm{net}}
=
Q_{\mathrm{het}}
+
Q_{\mathrm{inc}}
}.
}
\]
\rev{Thus, the established generalized Cochran \(Q\) decomposition has a direct Pythagorean interpretation in the CSP geometry: \(Q_{\mathrm{het}}\) is the residual distance from the design-specific fitted space, whereas \(Q_{\mathrm{inc}}\) is the additional distance introduced by restricting the design-specific fit to the globally consistent NMA space. The \(Q\) statistics themselves are established NMA diagnostics; CSP provides an equivalent nested-projection representation of their decomposition. Importantly, \(Q_{\mathrm{het}}\) and \(Q_{\mathrm{inc}}\) are complementary components of the same global statistic rather than separate competing analyses, and their sum reproduces \(Q_{\mathrm{net}}\) exactly up to numerical precision.}

\rev{Let \(r=\operatorname{rank}(\mathbf V)\). Under the fixed-effects model and the usual regularity conditions, the corresponding degrees of freedom are}
\[
\rev{
df_{\mathrm{net}}
=
r-(T-1),
\qquad
df_{\mathrm{het}}
=
r-p_D,
\qquad
df_{\mathrm{inc}}
=
p_D-(T-1),
}
\]
\rev{so that}
\[
\rev{
df_{\mathrm{net}}
=
df_{\mathrm{het}}
+
df_{\mathrm{inc}}.
}
\]
\rev{The three statistics may therefore be assessed using their corresponding chi-squared reference distributions. \(Q_{\mathrm{net}}\) provides a global test of the homogeneous consistency-constrained network model, \(Q_{\mathrm{het}}\) tests heterogeneity among studies sharing the same design, and \(Q_{\mathrm{inc}}\) tests inconsistency between design-specific treatment effects. When a component has zero degrees of freedom, no corresponding test is available.}

\rev{The connection with classical pairwise meta-analysis follows immediately. When \(T=2\), all studies have the same two-treatment design, so \(p_D=T-1=1\), \(\widehat{\mathbf y}_{\mathrm D}=\widehat{\mathbf y}_{\mathrm C}\), and therefore}
\[
\rev{
Q_{\mathrm{inc}}=0,
\qquad
df_{\mathrm{inc}}=0.
}
\]
\rev{For \(K\) independent studies,}
\[
\rev{
Q_{\mathrm{net}}
=
Q_{\mathrm{het}}
=
\sum_{i=1}^K
\frac{(y_i-\widehat\theta)^2}{\sigma_i^2},
\qquad
df_{\mathrm{net}}
=
df_{\mathrm{het}}
=
K-1,
}
\]
\rev{which is exactly Cochran's \(Q\) test for heterogeneity in pairwise meta-analysis. Thus, inconsistency is absent when only two treatments are present, while the global network statistic reduces to the classical heterogeneity statistic.}

\rev{The distinction between \(Q_{\mathrm{het}}\) and \(Q_{\mathrm{inc}}\) is design-based and should not be interpreted as an absolute separation of biological sources of variation. In particular, a design represented by only one study has no within-design heterogeneity degrees of freedom, so a deviation associated with that study may instead contribute to \(Q_{\mathrm{inc}}\). Accordingly, heterogeneity and inconsistency should be considered jointly \parencite{Krahn2013}.}

\rev{The formal decomposition and chi-squared tests above are presented for the fixed-effects specification used in the empirical applications. Under a random-effects NMA, between-study heterogeneity is modeled directly through \(\tau^2\) and incorporated in \(\mathbf V(\widehat{\tau}^2)\). Conditional on this fitted covariance matrix, the same nested projections yield the exact algebraic identity \(Q_{\mathrm{net}}^{\mathrm{RE}}=Q_{\mathrm{error}}+Q_{\mathrm{inc}}^{\mathrm{RE}}\), where \(Q_{\mathrm{error}}\) is the residual within-design error after accounting for the modeled heterogeneity and \(Q_{\mathrm{inc}}^{\mathrm{RE}}\) is the remaining between-design inconsistency component conditional on \(\widehat{\tau}^2\). Because the heterogeneity variance is estimated and already incorporated in the model, \(Q_{\mathrm{error}}\) should not be interpreted as a second heterogeneity test, and the fixed-effects chi-squared reference tests are not applied to these random-effects diagnostic components. For formal inconsistency inference under random effects, an established approach is to fit a random-effects design-by-treatment interaction model and jointly test the additional inconsistency parameters, for example using a global Wald test \parencite{Higgins2012,White2012}.}

\rev{\textbf{Running example continued.} The three studies have designs \(ABC\), \(AB\), and \(AC\), respectively. Hence}
\[
\rev{
p_D=(3-1)+(2-1)+(2-1)=4=\operatorname{rank}(\mathbf V).
}
\]
\rev{Because each design occurs only once, the design-specific model fits the observed contrasts exactly, giving}
\[
\rev{
Q_{\mathrm{het}}=0,
\qquad
df_{\mathrm{het}}=0.
}
\]
\rev{The global statistic is}
\[
\rev{
Q_{\mathrm{net}}=\frac{2}{3},
\qquad
df_{\mathrm{net}}=4-(3-1)=2.
}
\]
\rev{Therefore,}
\[
\rev{
Q_{\mathrm{inc}}
=
Q_{\mathrm{net}}-Q_{\mathrm{het}}
=
\frac{2}{3},
\qquad
df_{\mathrm{inc}}
=
4-(3-1)
=
2,
}
\]
\rev{corresponding to \(p=0.717\). In this example all estimable departure from the homogeneous consistency model is between designs because none of the designs is replicated, but the inconsistency test is not statistically significant.}

\rev{All of these CSP-based tools—the canonical forest plot, path-based visualizations, tension plots, and the generalized Cochran \(Q\) diagnostics—are illustrated in Section~\ref{sec:application}, demonstrating how the projection-based framework provides complementary views of evidence contribution, agreement between evidence sources, within-design heterogeneity, and between-design inconsistency. Specifically, the COVID-19 example illustrates canonical weights and multiple study-labeled paths in Figure~\ref{fig:csp_3d_decomposition}, Table~\ref{tab:covid_canonical_paths}, and Figure~\ref{fig:AE_paths_tikz}; the psoriasis example illustrates the tension plot in Figure~\ref{fig:psoriasis_tension}; the antidepressant example illustrates the study-level forest plot in Figure~\ref{fig:bupropion_placebo_forest}; and Table~\ref{tab:network_summary} reports \(Q_{\mathrm{net}}\), \(Q_{\mathrm{het}}\), and \(Q_{\mathrm{inc}}\) for all three empirical networks.}

\section{Applications}\label{sec:application}

To evaluate the practical performance of the proposed contrast-space
framework, we apply the method to three empirical \rev{NMA} datasets spanning markedly different levels of network
complexity. The examples are chosen to illustrate complementary aspects
of the methodology.

First, we analyze a compact COVID-19 therapeutic network constructed
from a subset of randomized platform trials included in the WHO living
NMA of COVID-19 treatments
\parencite{siemieniuk2020drug}, which informed early WHO treatment
guidelines \parencite{WHO_REACT_2020, WHO_Guideline_2020}. This example
provides a transparent setting in which the canonical direct and
indirect decomposition can be visualized explicitly and interpreted
at the individual study and path levels.

Second, we consider a moderately sized network of biologic therapies
for plaque psoriasis derived from randomized trials evaluating
PASI-75 response outcomes \parencite{Sbidian2021}. \rev{This network contains a mixture of two-arm and multi-arm studies and therefore provides a realistic example of the multi-arm covariance structures frequently encountered in contemporary clinical evidence synthesis.}

Finally, to demonstrate scalability, we analyze the large
antidepressant network compiled by \textcite{Cipriani2018}. This
dataset includes hundreds of trials and dozens of treatments,
generating a high-dimensional evidence structure in which indirect
comparisons are propagated through a complex web of active-comparator
trials.

Figure~\ref{fig:network_topologies} provides an overview of the
treatment-network structures for the three empirical datasets. The
networks differ substantially in both size and connectivity. The
COVID-19 example (Panel~A) represents a compact network dominated by
large multi-arm platform trials with shared control arms, resulting in
strong within-study covariance structures. The psoriasis biologics
network (Panel~B) is moderately sized and contains a mixture of
two-arm and multi-arm trials, producing a richer set of indirect
comparisons across treatments. In contrast, the antidepressant
network (Panel~C) forms a large and densely connected evidence
structure involving dozens of treatments and hundreds of trials.

\begin{figure}[htbp]
    \centering
    \includegraphics[height=0.85\textheight, keepaspectratio]{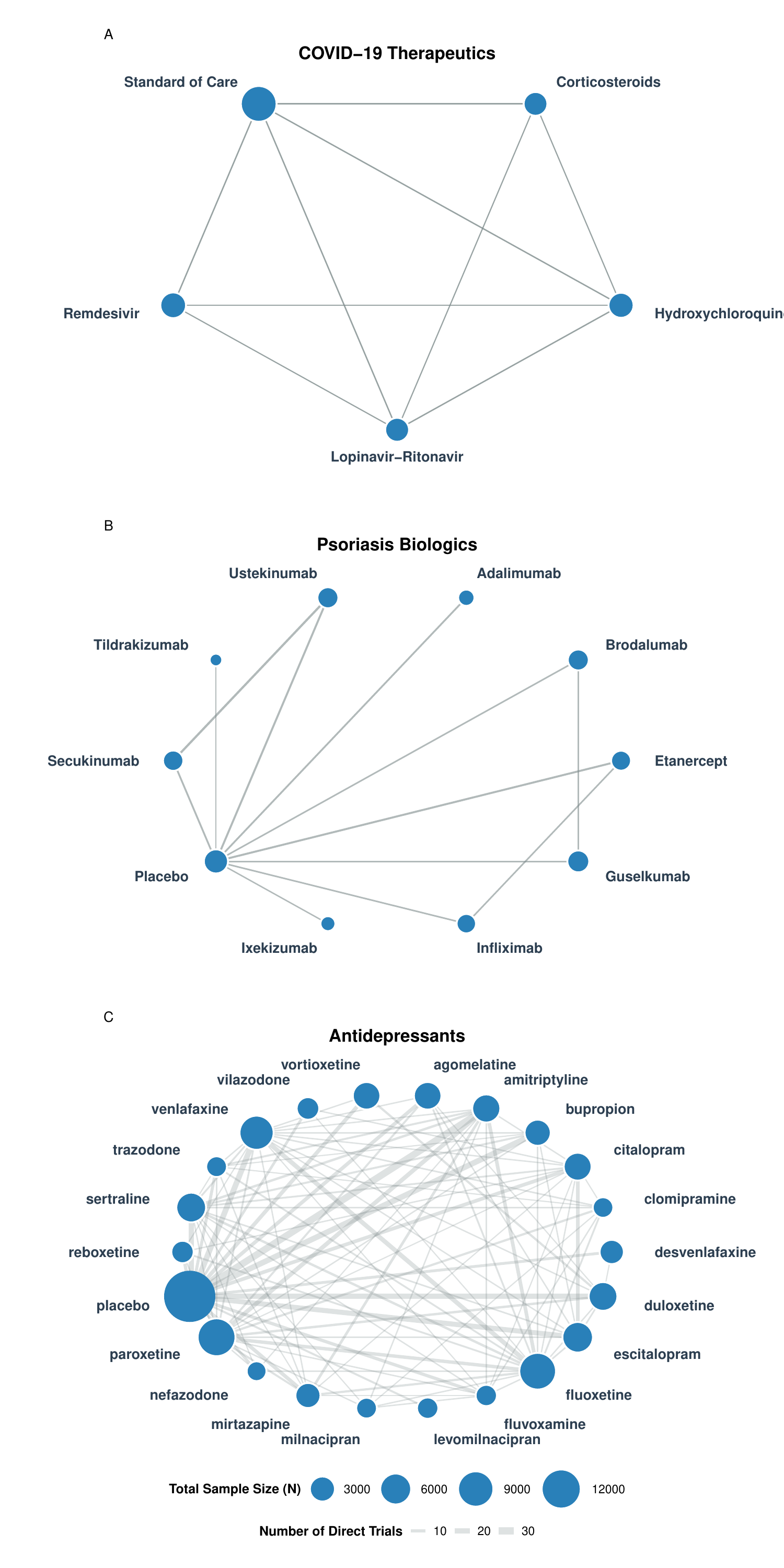}
    \caption{\textbf{Empirical network topologies.} Structural representation of the three clinical datasets: (A) COVID-19 therapeutics, (B) psoriasis biologics, and (C) the antidepressant network. Treatments are arranged radially to visualize topological density and indirect connectivity. Node area is proportional to the total number of randomized patients for that treatment. Edges represent direct head-to-head evidence, with thickness mapped globally to the number of independent trials. To preserve structural clarity in high-density graphs, edge transparency is dynamically scaled. A single shared global scale is used across panels}
    \label{fig:network_topologies}
\end{figure}

Together, these networks provide complementary empirical settings for
illustrating the proposed framework and the associated diagnostic tools
derived from the projection-based decomposition: a small pedagogical network for transparent study-level and path-level interpretation, a
medium-sized clinical network typical of modern comparative
effectiveness reviews, and a large-scale network illustrating the application of the approach in complex settings. All datasets were reconstructed from published sources and are used
here solely to illustrate the methodological properties of the proposed
framework.

\subsection{Overview of the empirical networks}

Table~\ref{tab:network_summary} summarizes the key structural
characteristics of the three empirical networks analyzed in this
section. The examples represent increasing levels of network
complexity: a compact COVID-19 treatment network constructed from
major platform trials, a moderately sized biologics network for
psoriasis containing a mixture of two-arm and multi-arm studies, and the large antidepressant network. \rev{All analyses presented here are conducted under a fixed-effects specification. Under a random-effects specification, the fitted covariance matrix becomes \(\mathbf V(\widehat{\tau}^2)\), while the NMA estimator and the CSP direct, indirect, study-level, and path-level decompositions retain the same algebraic construction conditional on that covariance matrix. The interpretation of the \(Q\) diagnostics differs under random effects, as described in Section~\ref{sec:qdecomp}.}

\begin{table}[ht]
\centering
\caption{\rev{Generalized Cochran \(Q\) diagnostics for the empirical NMA datasets}}
\label{tab:network_summary}
\resizebox{\linewidth}{!}{
\begin{tabular}{lccccc}
\toprule
Dataset & \#Treatments (\(T\)) & rank(\(\mathbf V\)) & \(Q_{\mathrm{net}}\) (\(df\); \(p\)) & \(Q_{\mathrm{het}}\) (\(df\); \(p\)) & \(Q_{\mathrm{inc}}\) (\(df\); \(p\)) \\
\midrule
COVID therapeutics       & 5  & 9   & 5.687 (5; 0.338) & 0.357 (1; 0.550) & 5.330 (4; 0.255) \\
Psoriasis biologics      & 10 & 33  & 28.578 (24; 0.236) & 25.688 (20; 0.176) & 2.890 (4; 0.576) \\
Cipriani antidepressants & 22 & 542 & 760.236 (521; \(<0.001\)) & 526.763 (380; \(<0.001\)) & 233.473 (141; \(<0.001\)) \\
\bottomrule
\end{tabular}
}
\end{table}

\rev{In addition to summarizing network size, Table~\ref{tab:network_summary} reports the generalized Cochran \(Q\) diagnostics for each empirical network. The global statistic satisfies \(Q_{\mathrm{net}}=Q_{\mathrm{het}}+Q_{\mathrm{inc}}\), where \(Q_{\mathrm{het}}\) measures heterogeneity among studies sharing the same design and \(Q_{\mathrm{inc}}\) measures inconsistency between designs. The corresponding degrees of freedom satisfy the same decomposition. Thus, these diagnostics use the same observed contrasts and covariance structure, with the design-specific projection serving only to partition the global statistic into its within-design heterogeneity and between-design inconsistency components; it does not produce an alternative set of network treatment-effect estimates.}

\rev{For the COVID-19 network, neither the global statistic (\(Q_{\mathrm{net}}=5.69\), \(df=5\), \(p=0.338\)) nor its heterogeneity (\(Q_{\mathrm{het}}=0.36\), \(df=1\), \(p=0.550\)) or inconsistency (\(Q_{\mathrm{inc}}=5.33\), \(df=4\), \(p=0.255\)) component is statistically significant. The psoriasis network likewise shows no significant global departure (\(Q_{\mathrm{net}}=28.58\), \(df=24\), \(p=0.236\)), within-design heterogeneity (\(Q_{\mathrm{het}}=25.69\), \(df=20\), \(p=0.176\)), or between-design inconsistency (\(Q_{\mathrm{inc}}=2.89\), \(df=4\), \(p=0.576\)). In contrast, the antidepressant network shows strong evidence of both within-design heterogeneity (\(Q_{\mathrm{het}}=526.76\), \(df=380\), \(p<0.001\)) and between-design inconsistency (\(Q_{\mathrm{inc}}=233.47\), \(df=141\), \(p<0.001\)), which together produce the significant global statistic \(Q_{\mathrm{net}}=760.24\) with \(df=521\).}

\subsection{COVID-19 therapeutic network}

We first consider a five-treatment network based on several
large randomized platform trials evaluating COVID-19 therapies. The treatments include standard care, corticosteroids, lopinavir--ritonavir, hydroxychloroquine, and remdesivir.

To demonstrate the proposed contrast-space framework in a transparent
setting, we constructed a five-treatment subnetwork from the living
NMA of COVID-19 therapeutics
\parencite{siemieniuk2020drug}. We focus on a core cluster of
foundational platform trials evaluating corticosteroids and antiviral
therapies (RECOVERY, SOLIDARITY, ACTT-1, CoDEX, and REMAP-CAP).
Clinically, the intersection of these trials formed the primary
evidence base for the World Health Organization's early guidance on
systemic COVID-19 therapies
\parencite{WHO_REACT_2020, WHO_Guideline_2020}. This subset provides a
tractable example in which the covariance structure induced by
multi-arm platform trials can be verified explicitly, and the
resulting projection-based decomposition can be interpreted at the
individual study and path levels. This example is used solely for methodological illustration and does not aim to provide a complete or up-to-date synthesis of the COVID-19 evidence base.

Panel~A of Figure~\ref{fig:network_topologies} displays the treatment
network. Because several platform trials simultaneously evaluated
multiple treatments against a shared control arm, the covariance
matrix of the direct estimates contains non-zero off-diagonal elements
reflecting within-study correlation induced by shared controls in
multi-arm trials.

\rev{For each target comparison, the canonical decomposition assigns weights to the direct study contributions and to the aggregate indirect component, with the weights summing to one. Figure~\ref{fig:csp_3d_decomposition} provides a graphical overview of these weights across all target comparisons in the COVID-19 network, allowing the relative importance of direct studies and indirect evidence to be compared visually across the network. The purpose of this display is complementary to Table~\ref{tab:covid_canonical_paths}: the figure provides an overall visual summary for practitioners, whereas the table reports the exact numerical weights and the individual study-labeled paths underlying each comparison.}

\begin{figure}[htbp]
\centering
\includegraphics[
    width=\textwidth,
    trim=4 0 0 0, clip, 
    height=1\textheight,
    keepaspectratio
]{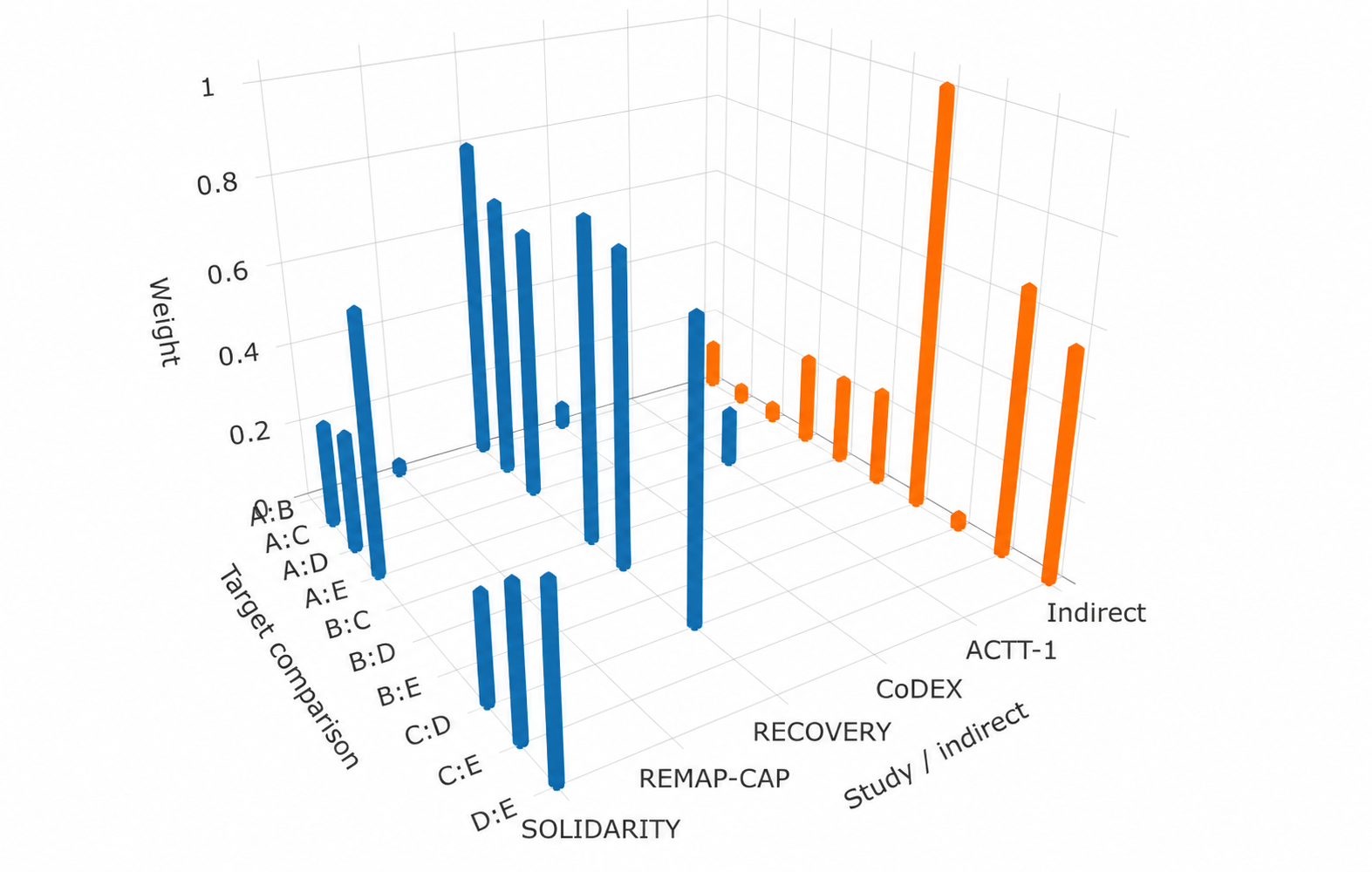}
\caption{\textbf{Three-dimensional display of the decomposition with canonical weights in the COVID-19 network.}
For each target comparison, vertical bars represent the exact weights
in the canonical CSP decomposition. Along the horizontal axes, one
dimension indexes the target comparison and the other indexes the
contributing source. Blue bars correspond to direct study
contributions, labeled by study, and the orange bar gives the total
indirect weight obtained by summing over all canonical indirect
paths. \rev{For each target comparison, the displayed weights sum to one. The figure is intended to provide a visual overview of contribution patterns across comparisons; exact numerical weights and individual indirect paths are reported in Table~\ref{tab:covid_canonical_paths}}}
\label{fig:csp_3d_decomposition}
\end{figure}

The decomposition highlights substantial variation across target comparisons. Some contrasts are dominated by direct evidence from one
or two studies, whereas others depend more heavily on indirect
network pathways. In particular, the comparison between
corticosteroids and remdesivir (\(B{:}E\)) is informed entirely by
indirect evidence, while comparisons such as \(A{:}B\), \(A{:}C\), and
\(B{:}C\) retain dominant direct components.

Table~\ref{tab:covid_canonical_paths} lists the full canonical
direct and indirect study-path decomposition for the COVID-19
network. For each target comparison, the table reports the canonical direct
study paths, the canonical indirect paths, and their exact weights in
the linear decomposition of the corresponding network estimate. These
weights sum to one within each target comparison and therefore provide
a complete coefficient-scale representation of how evidence is
distributed across studies and pathways.

\begin{table}[htbp]
\centering
\caption{Canonical direct and indirect study-path decomposition for the COVID-19 network. For each target comparison, the table lists the canonical direct study paths and indirect paths identified by the CSP decomposition. \rev{The weight column gives the canonical coefficient of each path in the linear decomposition of the corresponding network estimate. The underlying weights sum to one within each target comparison; displayed weights are rounded to four decimal places and therefore may not sum exactly to one.} \textit{Abbreviations: Dir = Direct; Ind = Indirect; R = RECOVERY; S = SOLIDARITY; A = ACTT-1; C = CoDEX; M = REMAP-CAP}}
\label{tab:covid_canonical_paths}
\vspace{0.5em}

\scriptsize
\setlength{\tabcolsep}{2pt}

\makebox[\textwidth][c]{%
\begin{tabular}[t]{@{} l l l l l @{}}
\toprule
Vs. & Type & Studies & Path & Weight\\
\midrule
A:B & Dir & R & $A{\to}B$ & 0.7923 \\
    & Dir & C & $A{\to}B$ & 0.0616 \\
    & Dir & M & $A{\to}B$ & 0.0343 \\
    & Ind & S/R & $A{\to}D{\to}B$ & 0.0636 \\
    & Ind & S/R & $A{\to}C{\to}B$ & 0.0412 \\
    & Ind & A/S/R & $A{\to}E{\to}C{\to}B$ & 0.0070 \\
\midrule
A:C & Dir & R & $A{\to}C$ & 0.6971 \\
    & Dir & S & $A{\to}C$ & \rev{0.2632} \\
    & Ind & C/R & $A{\to}B{\to}C$ & 0.0143 \\
    & Ind & A/S & $A{\to}E{\to}C$ & \rev{0.0015} \\
    & Ind & A/S/R & $A{\to}E{\to}D{\to}C$ & \rev{0.0161} \\
    & Ind & M/R & $A{\to}B{\to}C$ & 0.0079 \\
\midrule
A:D & Dir & R & $A{\to}D$ & 0.6618 \\
    & Dir & S & $A{\to}D$ & 0.2972 \\
    & Ind & A/S & $A{\to}E{\to}D$ & 0.0167 \\
    & Ind & C/R & $A{\to}B{\to}D$ & 0.0135 \\
    & Ind & M/R & $A{\to}B{\to}D$ & 0.0075 \\
    & Ind & A/S/R & $A{\to}E{\to}C{\to}D$ & 0.0032 \\
\midrule
A:E & Dir & S & $A{\to}E$ & 0.6392 \\
    & Dir & A & $A{\to}E$ & 0.1405 \\
    & Ind & R/S & $A{\to}D{\to}E$ & 0.1253 \\
    & Ind & R/S & $A{\to}C{\to}E$ & 0.0882 \\
    & Ind & C/R/S & $A{\to}B{\to}C{\to}E$ & 0.0044 \\
    & Ind & M/R/S & $A{\to}B{\to}C{\to}E$ & 0.0024 \\
\midrule
B:C & Dir & R & $B{\to}C$ & 0.7835 \\
    & Ind & R/S & $B{\to}A{\to}C$ & \rev{0.0953} \\
    & Ind & R/S & $B{\to}D{\to}C$ & 0.0475 \\
    & Ind & C/S & $B{\to}A{\to}C$ & 0.0473 \\
    & Ind & M/S & $B{\to}A{\to}C$ & \rev{0.0158} \\
    & Ind & \rev{M/A/S} & $B{\to}A{\to}E{\to}C$ & 0.0106 \\
\bottomrule
\end{tabular}%
\hspace{1em}
\begin{tabular}[t]{@{} l l l l l @{}}
\toprule
Vs. & Type & Studies & Path & Weight \\
\midrule
B:D & Dir & R & $B{\to}D$ & 0.7496 \\
    & Ind & R/S & $B{\to}A{\to}D$ & \rev{0.1305} \\
    & Ind & C/S & $B{\to}A{\to}D$ & 0.0481 \\
    & Ind & R/S & $B{\to}C{\to}D$ & 0.0450 \\
    & Ind & M/S & $B{\to}A{\to}D$ & \rev{0.0139} \\
    & Ind & \rev{M/A/S} & $B{\to}A{\to}E{\to}D$ & 0.0129 \\
\midrule
B:E & Ind & R/S & $B{\to}A{\to}E$ & \rev{0.5344} \\
    & Ind & R/S & $B{\to}D{\to}E$ & 0.1889 \\
    & Ind & R/S & $B{\to}C{\to}E$ & 0.1432 \\
    & Ind & R/A & $B{\to}A{\to}E$ & \rev{0.0444} \\
    & Ind & \rev{C/A} & $B{\to}A{\to}E$ & 0.0572 \\
    & Ind & M/A & $B{\to}A{\to}E$ & \rev{0.0319} \\
\midrule
C:D & Dir & R & $C{\to}D$ & 0.7021 \\
    & Dir & S & $C{\to}D$ & 0.2615 \\
    & Ind & R/S & $C{\to}A{\to}D$ & 0.0341 \\
    & Ind & R/A/S & $C{\to}A{\to}E{\to}D$ & 0.0012 \\
    & Ind & R/C/A/S & $C{\to}B{\to}A{\to}E{\to}D$ & 0.0007 \\
    & Ind & R/M/A/S & $C{\to}B{\to}A{\to}E{\to}D$ & 0.0004 \\
\midrule
C:E & Dir & S & $C{\to}E$ & 0.3597 \\
    & Ind & R/S & $C{\to}A{\to}E$ & \rev{0.3760} \\
    & Ind & R/S & $C{\to}D{\to}E$ & 0.1414 \\
    & Ind & R/A & $C{\to}A{\to}E$ & \rev{0.1075} \\
    & Ind & R/C/A & $C{\to}B{\to}A{\to}E$ & \rev{0.0099} \\
    & Ind & \rev{R/M/A} & $C{\to}B{\to}A{\to}E$ & 0.0055 \\
\midrule
D:E & Dir & S & $D{\to}E$ & 0.4392 \\
    & Ind & R/S & $D{\to}A{\to}E$ & \rev{0.3419} \\
    & Ind & R/A & $D{\to}A{\to}E$ & \rev{0.1064} \\
    & Ind & R/S & $D{\to}C{\to}E$ & 0.0982 \\
    & Ind & R/C/A & $D{\to}B{\to}A{\to}E$ & 0.0092 \\
    & Ind & R/M/A & $D{\to}B{\to}A{\to}E$ & \rev{0.0051} \\
\bottomrule
\end{tabular}%
}
\end{table}

As one concrete example, Figure~\ref{fig:AE_paths_tikz} displays the
canonical direct and indirect study paths for the target comparison
between standard care and remdesivir (\(A{:}E\)). The two dominant
direct contributions arise from SOLIDARITY and ACTT-1, with weights
0.6392 and 0.1405, respectively. The remaining weight is partitioned
among four indirect paths, the largest of which is the
RECOVERY--SOLIDARITY path \(A \to D \to E\) with weight 0.1253. This
visualization makes clear that even in a small network, the network
estimate may be assembled from a mixture of direct evidence and
multiple lower-weight indirect pathways.

\rev{The \(A{:}E\) comparison also illustrates the removal of apparent cycles. Before canonicalization, the raw full-contrast coefficients included a RECOVERY segment \(C\to D\) and a SOLIDARITY segment \(D\to C\), each with weight approximately \(0.0076\). Pooling these raw segments produced an apparent \(C\to D\to C\) cycle. The within-study treatment-balance reduction replaced these algebraically redundant internal segments by equivalent positive-to-negative endpoint transfers. Consequently, neither opposing segment remained in the canonical graph, and the subsequent path extraction left no positive residual coefficient or directed cycle.}

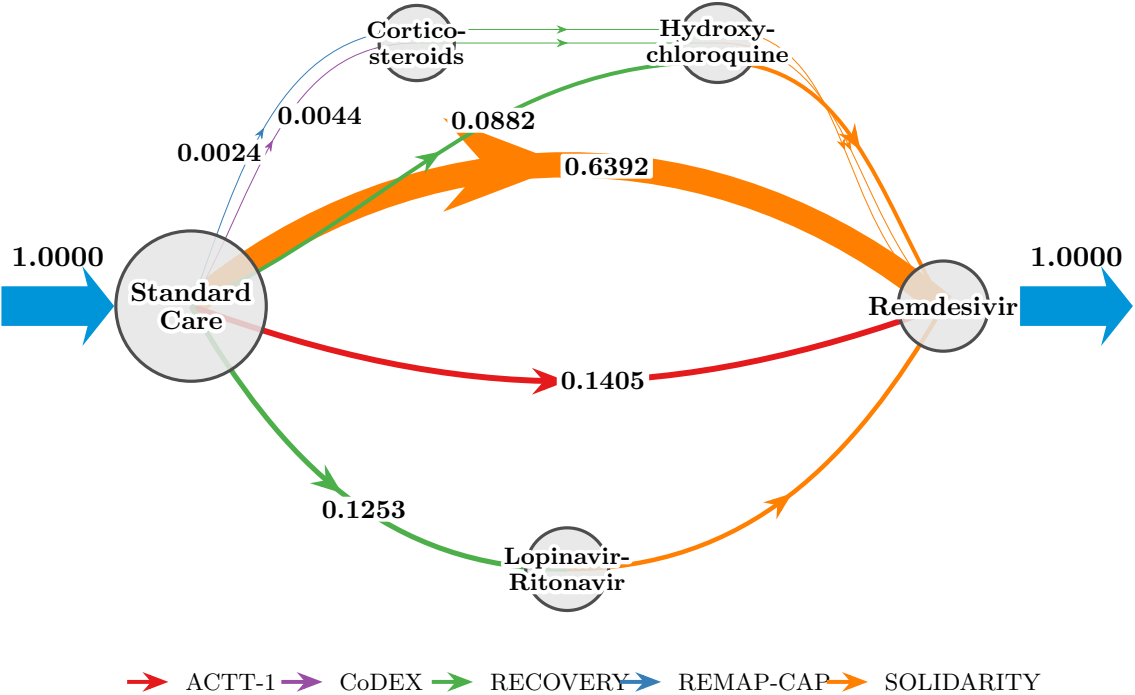
\begin{figure}[ht]
\centering
\begin{tikzpicture}[
    >=Stealth,
    node_base/.style={circle, draw=black!70, line width=1.2pt, inner sep=0pt, fill=gray!20, fill opacity=0.85},
    wt/.style={fill=white, inner sep=1.5pt, rounded corners=1pt, execute at begin node={\small\bfseries}, text=black},
    mid arrow/.style={postaction={decorate, decoration={markings, mark=at position 0.5 with {\arrow{>}}}}}
]

\pgfmathsetmacro{\maxflowwidth}{15} 

\definecolor{actt}{HTML}{E41A1C}
\definecolor{codex}{HTML}{984EA3}
\definecolor{recovery}{HTML}{4DAF4A}
\definecolor{remap}{HTML}{A65628}
\definecolor{solidarity}{HTML}{FF7F00}

\coordinate (A) at (0.0,  0.0);
\coordinate (E) at (10.0, 0.0);
\coordinate (B) at (3.0,  3.5);
\coordinate (C) at (7.0,  3.5);
\coordinate (D) at (5.0, -3.5);

\coordinate (Bmid) at (B);
\coordinate (Btop) at ($ (B) + (0, 1.8mm) $);
\coordinate (Cmid) at (C);
\coordinate (Ctop) at ($ (C) + (0, 1.8mm) $);
\coordinate (Cbot) at ($ (C) - (0, 2.5mm) $);

\draw[solidarity,  mid arrow,  line width=0.6392*\maxflowwidth pt] (A) to[out=40, in=140]
    node[pos=0.55, wt] {0.6392} (E);

\draw[actt,  mid arrow,  line width=0.15*\maxflowwidth pt] (A) to[out=-20, in=200]
    node[pos=0.55, wt] {0.1405} (E);

\draw[recovery,  mid arrow, line width=0.13*\maxflowwidth pt] (A) to[out=-60, in=180]
    node[pos=0.55, wt] {0.1253} (D);
\draw[solidarity,  mid arrow,  line width=1.7pt] (D) to[out=0, in=240] (E);

\draw[recovery,  mid arrow, line width=0.1*\maxflowwidth pt] (A) to[out=25, in=180]
    node[pos=0.6, wt] {0.0882} (Cbot);
\draw[solidarity,  mid arrow,  line width=0.1*\maxflowwidth pt] (Cbot) to[out=0, in=120] (E);

\draw[codex, mid arrow, line width=0.8pt] (A) to[out=60, in=180] 
    node[pos=0.55, xshift=12pt, wt] {0.0044} (Bmid);

\draw[recovery, mid arrow, line width= 0.01*\maxflowwidth pt] (Bmid) to[out=0, in=180] (Cmid);
\draw[solidarity, mid arrow, line width=0.8pt] (Cmid) to[out=0, in=135] (E);

\draw[remap,  mid arrow, line width= 0.01*\maxflowwidth pt] (A) to[out=70, in=180]
    node[pos=0.4,xshift=-12pt,  wt] {0.0024} (Btop);
\draw[recovery,  mid arrow, line width=0.8pt] (Btop) to[out=0, in=180] (Ctop);
\draw[solidarity,  mid arrow,  line width=0.8pt] (Ctop) to[out=0, in=145] (E);

\node[node_base, minimum size=2.00cm] (nodeA) at (A) {};
\node[align=center, font=\small\bfseries] at (A) {\contour{white}{Standard}\\[-0.3ex]\contour{white}{Care}};

\node[node_base, minimum size=1.00cm] (nodeB) at (B) {};
\node[align=center, font=\footnotesize\bfseries] at (B) {\contour{white}{Cortico-}\\[-0.3ex]\contour{white}{steroids}};

\node[node_base, minimum size=1.05cm] (nodeC) at (C) {};
\node[align=center, font=\footnotesize\bfseries] at (C) {\contour{white}{Hydroxy-}\\[-0.3ex]\contour{white}{chloroquine}};

\node[node_base, minimum size=1.10cm] (nodeD) at (D) {};
\node[align=center, font=\footnotesize\bfseries] at (D) {\contour{white}{Lopinavir-}\\[-0.3ex]\contour{white}{Ritonavir}};

\node[node_base, minimum size=1.20cm] (nodeE) at (E) {};
\node[align=center, font=\small\bfseries] at (E) {\contour{white}{Remdesivir}};

\draw[-{Stealth[length=12pt, width=30pt]}, line width=\maxflowwidth pt, cyan!80!blue] ($(nodeA.west) + (-1.5, 0)$) -- ($(nodeA.west) + (0, 0)$)
    node[midway, above = 4pt, text=black] {\textbf{1.0000}};
\draw[-{Stealth[length=12pt, width=30pt]}, line width= \maxflowwidth pt, cyan!80!blue] ($(nodeE.east) + (0.4, 0)$) -- ($(nodeE.east) + (1.9, 0)$)
    node[midway, above = 4pt, text=black] { \textbf{1.0000}};

\begin{scope}[shift={(-0.85, -5)}]
    \draw[actt, ->, line width=1.5pt] (0.0,0) -- (0.55,0)
        node[anchor=west, xshift=2pt, text=black] {\footnotesize ACTT-1};

    \draw[codex, ->, line width=1.5pt] (2.05,0) -- (2.60,0)
        node[anchor=west, xshift=2pt, text=black] {\footnotesize CoDEX};

    \draw[recovery, ->, line width=1.5pt] (4.05,0) -- (4.60,0)
        node[anchor=west, xshift=2pt, text=black] {\footnotesize RECOVERY};

    \draw[remap, ->, line width=1.5pt] (6.55,0) -- (7.10,0)
        node[anchor=west, xshift=2pt, text=black] {\footnotesize REMAP-CAP};

    \draw[solidarity, ->, line width=1.5pt] (9.3,0) -- (9.85,0)
        node[anchor=west, xshift=2pt, text=black] {\footnotesize SOLIDARITY};
\end{scope}

\end{tikzpicture}
\caption{\textbf{Canonical direct and indirect study paths for the target comparison \(A{:}E\) (standard care versus remdesivir) in the COVID-19 network.}
Each continuous path represents one component in the canonical path-level decomposition of the \(A{:}E\) network estimate. \rev{The arrows entering \(A\) and leaving \(E\), each labeled 1.0000, indicate that the canonical path weights sum to one. Within each path, segment colors identify the study contributing the corresponding treatment comparison; therefore, an indirect path involving multiple studies may change color along its route. Path widths are proportional to the corresponding canonical weights, with a minimum width imposed on the smallest paths to maintain visibility, and the displayed values indicate the path weights.} The two thick direct paths from \(A\) to \(E\) dominate the decomposition, while the remaining indirect paths are routed through other treatments in the network}
\label{fig:AE_paths_tikz}
\end{figure}

\subsection{Psoriasis biologics network}

We next analyze a network of biologic therapies for moderate-to-severe
plaque psoriasis derived from randomized clinical trials evaluating
PASI-75 response rates. \rev{The dataset contains a mixture of two-arm and multi-arm studies, making it a useful example for illustrating the handling of multi-arm correlations.}

The network includes 23 studies comparing 10 treatments and produces
43 observed study-level contrasts. 

Figure~\ref{fig:psoriasis_mosaic_direct_indirect} summarizes the direct
and indirect decomposition of comparison-level contributions for the psoriasis biologics network. Each comparison is represented by a horizontal band, where the partition between direct and indirect components shows how the network estimate is constructed from the two sources of evidence. This representation makes it possible to assess
which comparisons rely primarily on direct evidence, which are driven
by indirect pathways, and which reflect a mixture of both. In the
psoriasis network, several comparisons are dominated by indirect
contributions, whereas a smaller subset retains a substantial direct
component, reflecting variation in how evidence propagates across the treatment network.
\begin{figure}[t]
\centering
\includegraphics[width=1\textwidth]{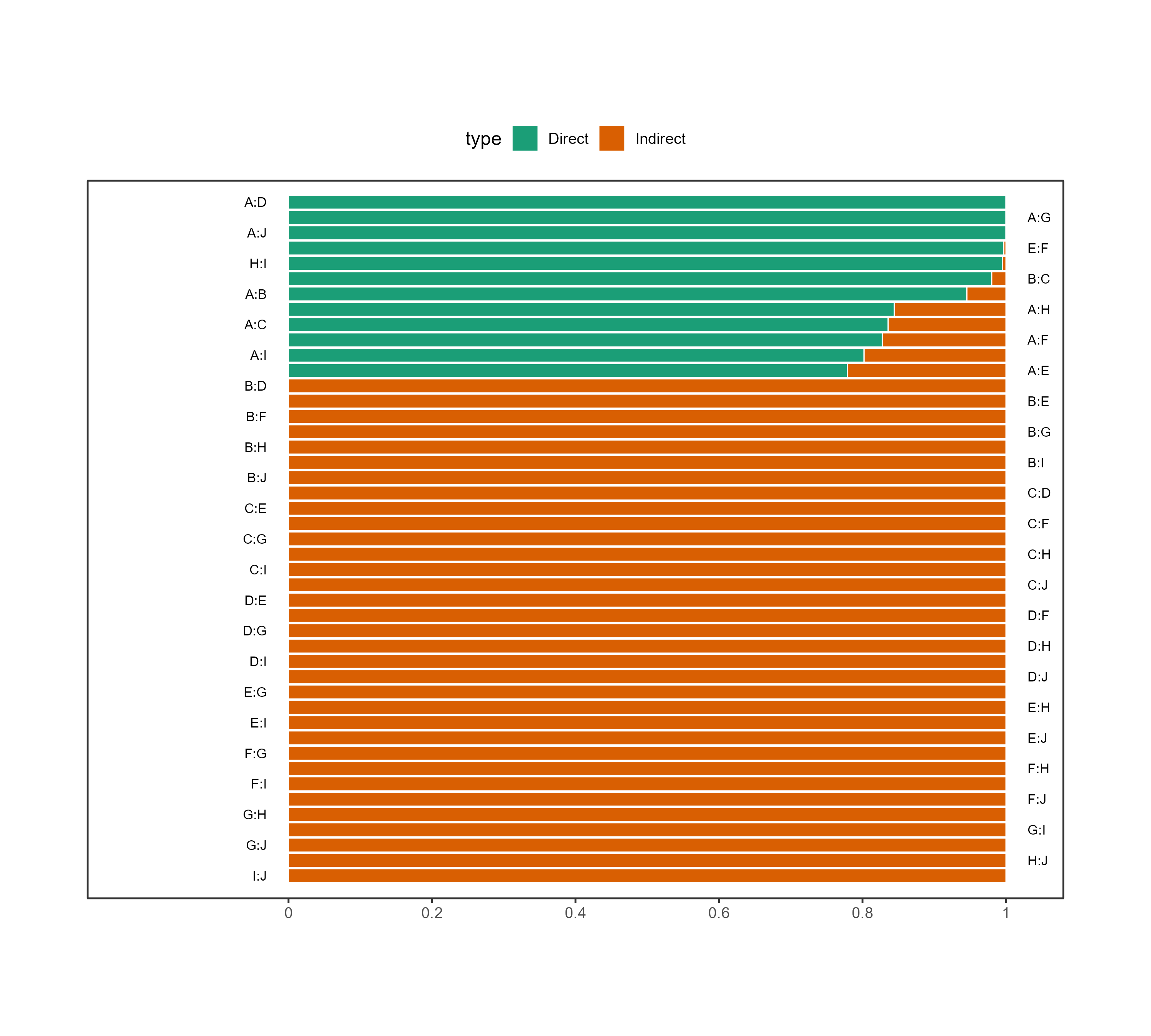}
\caption{
\textbf{Direct and indirect contribution decomposition for all pairwise comparisons in the psoriasis biologics network.}
Each horizontal band corresponds to a target comparison \(a{:}b\). Within each band, the width is normalized to one and partitioned into direct and indirect components according to their relative contribution weights. The horizontal split therefore represents how the network estimate for each comparison is composed from direct and indirect evidence}
\label{fig:psoriasis_mosaic_direct_indirect}
\end{figure}

To further examine agreement between evidence sources, we construct a tension plot based on the contrast-space projection decomposition (Figure~\ref{fig:psoriasis_tension}). For each comparison against the baseline treatment \(A\), the network estimate is decomposed into direct and indirect components on the estimator scale with associated uncertainty, and the figure directly compares these components with the overall network estimate. \rev{Because all quantities arise from the same linear decomposition, the three estimates are algebraically linked and fully consistent with the GLS model.}

\begin{figure}[htbp]
\centering
\includegraphics[width=\linewidth]{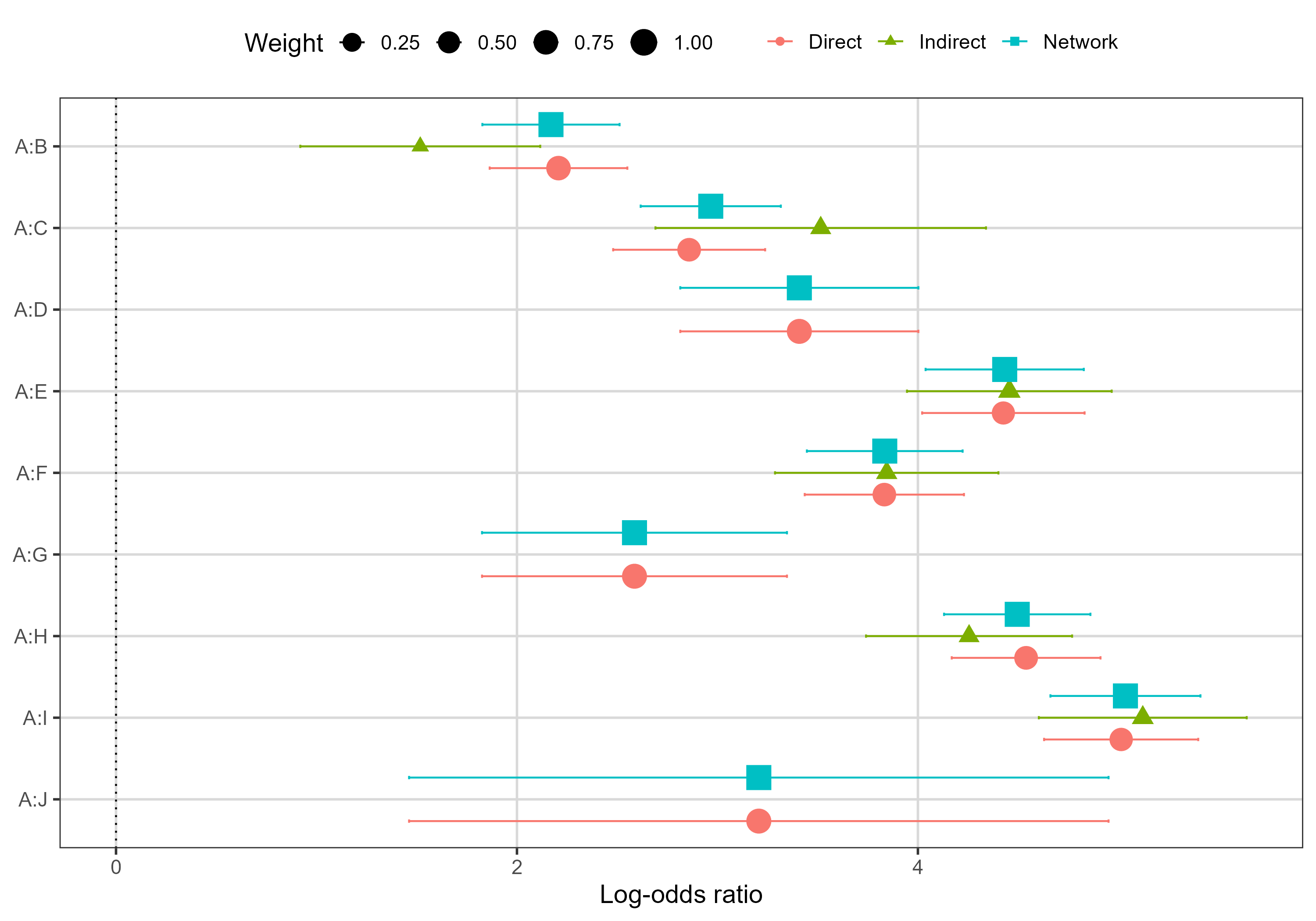}
\caption{
\textbf{Direct, indirect, and network estimates for baseline comparisons in the psoriasis biologics network.}
For each target comparison against treatment \(A\), points represent the direct, indirect, and network estimates derived from the contrast-space projection decomposition, with horizontal bars indicating 95\% confidence intervals. Point sizes are scaled according to the corresponding canonical direct and indirect weights. The vertical dotted line at zero corresponds to no treatment effect on the log-odds-ratio scale. \rev{Close agreement between direct and indirect components suggests consistency,} while discrepancies reflect tension between evidence sources}

\label{fig:psoriasis_tension}
\end{figure}

When direct and indirect estimates are closely aligned, the network estimate lies between them with reduced uncertainty, indicating coherent evidence integration; in contrast, separation between the direct and indirect components signals tension within the network, potentially arising from heterogeneity or inconsistency in the underlying evidence. 

\subsection{Antidepressant treatment network}

Finally, we apply the framework to the large antidepressant network
compiled by \textcite{Cipriani2018}. This network includes 473 trials
comparing 22 pharmacological treatments, forming a high-dimensional evidence structure with extensive
indirect connectivity. Because the network is extremely large, full visualization of all study--contrast or path-level contributions is difficult to interpret. We therefore focus on study-level and comparison-specific decompositions derived from the projection-based representation, as illustrated in Figure~\ref{fig:bupropion_placebo_forest} for the comparison between bupropion and placebo. \rev{The figure displays the direct study estimates and their canonical weights together with the aggregated direct and indirect estimates, allowing direct assessment of how individual studies and evidence sources contribute to the network estimate.}

\rev{In this example, the network estimate is $-0.412$, with $65.2\%$ direct evidence and $34.8\%$ indirect evidence.} The aggregated overall direct estimate ($-0.374$) and indirect estimate ($-0.483$) are both negative but differ in magnitude, with the indirect component shifting the network estimate toward a stronger treatment effect. The distribution of study-level contributions is highly uneven and strongly skewed: a small fraction of studies accounts for a disproportionately large share of the total weight, while the majority of studies contribute only marginally, forming a long tail of small contributions. \rev{Although several studies receive relatively larger weights, no single study dominates the overall estimate, and the remaining weight is distributed across a long tail of smaller contributions.} The confidence intervals reflect substantial variability across study-level contributions, including at least one study with a markedly more extreme estimate and limited overlap with others, indicating variation among the contributing study estimates.

\begin{figure}[t]
\centering
\includegraphics[width=\textwidth]{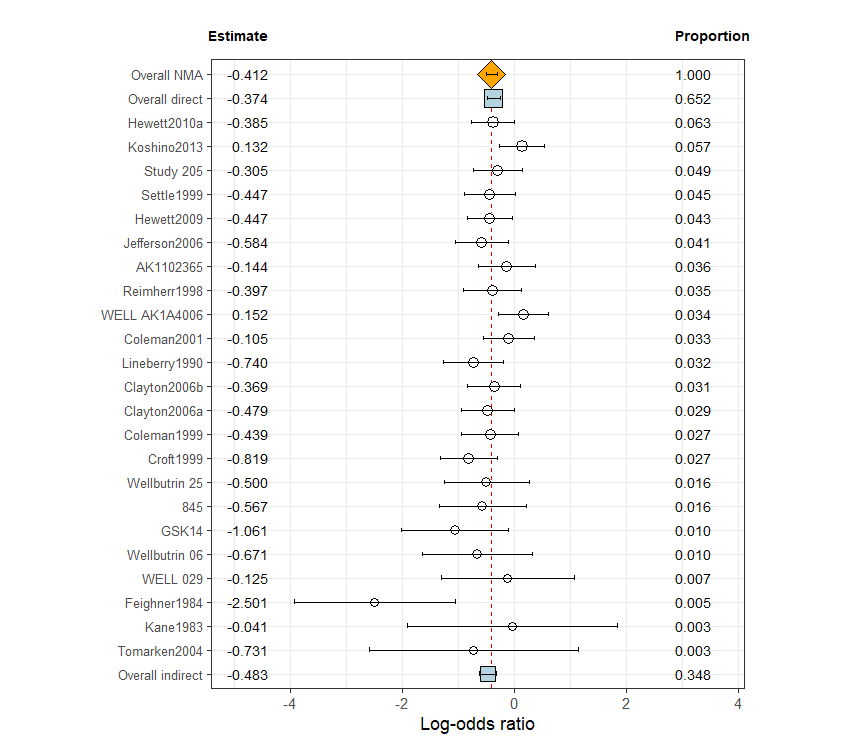}
\caption{
\textbf{Forest plot for canonical direct and indirect study contributions for the comparison bupropion versus placebo in the antidepressant network.}
The red dashed vertical line marks the overall NMA estimate. The two shaded summary rows display the aggregated direct and indirect components derived from the contrast-space projection decomposition. Each remaining row corresponds to an individual direct-study estimate, shown with its 95\% confidence interval. Point size is proportional to the corresponding canonical direct weight, and the values in the rightmost column report the corresponding weights}
\label{fig:bupropion_placebo_forest}
\end{figure}

\rev{All contributions shown in the figure arise from the same fitted CSP mapping, ensuring that the decomposition is numerically reproducible and fully consistent with the fitted network model.} Although the canonical decomposition resolves indirect evidence into path-level components, the number of such paths can be large in high-dimensional networks; in these settings, the study-level representation is more interpretable, while path-level decompositions are more informative in smaller networks such as the COVID-19 example. Together with the earlier visualizations, this example demonstrates how the contrast-space projection framework enables consistent interpretation of evidence from individual studies to network-wide structure within a unified linear-operator formulation.

\section{Conclusion and discussion}\label{sec:conclusion}

\rev{We present a reproducible linear framework for NMA in which the estimator, contribution measures, and diagnostic visualizations arise from the same GLS mapping of observed contrasts to fitted network estimates, while the generalized Cochran \(Q\) decomposition is represented through nested GLS projections using the same observed contrasts and covariance structure.} Building on this representation, we develop a covariance-aware and representation-invariant decomposition of the network estimator into direct, indirect, study-level, edge-level, and path-level contributions that remains algebraically consistent with the \rev{GLS} formulation of NMA. Across the three empirical applications, this framework provides a transparent account of how evidence propagates through the treatment network. Within this framework, the projection geometry provides the common mechanism that separates direct from indirect evidence, canonical from non-canonical study representations, and the global departure from the homogeneous consistency model into within-design heterogeneity and between-design inconsistency.

A central contribution of this work is the introduction of a rigorous
study-based definition of direct and indirect evidence through
canonical within-study reduction. This construction removes algebraic
redundancy induced by multi-arm trials and yields a unique and
invariant decomposition that does not depend on the choice of
within-study parameterization. As a result, study-level direct and
indirect contributions are defined unambiguously and are fully
reproducible across equivalent representations of the data.

The projection-based formulation also yields a natural generalization
of classical inverse-variance weighting. In the pairwise setting,
study weights determine how individual studies contribute to the
pooled estimate. In the network setting, the projection coefficients
play an analogous role, assigning interpretable weights to both
direct study contributions and indirect path-level components. This
connection provides a unifying perspective linking classical
meta-analysis and NMA within a common linear
framework.

Building on the CSP framework, we develop a set of diagnostic and
graphical tools—including study-based forest plots, \rev{the generalized Cochran \(Q\) diagnostics \(Q_{\mathrm{net}}\), \(Q_{\mathrm{het}}\), and \(Q_{\mathrm{inc}}\)}, tension plots comparing direct and indirect components, and path-based
visualizations—that provide interpretable summaries of how evidence
contributes to each network estimate. \rev{The contribution and visualization tools are derived directly from the fitted CSP mapping and therefore require no alternative NMA fit, while the generalized Cochran \(Q\) decomposition uses an auxiliary design-specific projection solely to partition \(Q_{\mathrm{net}}\) into its heterogeneity and inconsistency components. } In particular, the canonical decomposition enables the construction of study-level forest plots that exactly reconstruct the NMA estimate, with each study contributing additively through its canonical component. This yields a transparent and interpretable decomposition of the network estimate and provides, to our knowledge, the first direct extension of classical inverse-variance–weighted forest plots to the network setting, allowing clear interpretation of study contributions in complex treatment networks.

The empirical examples illustrate how the proposed framework enables
interpretation at multiple levels of resolution. In small networks,
such as the COVID-19 example, indirect evidence can be decomposed into
individual study paths, providing a detailed view of evidence flow. In
larger networks, such as the antidepressant dataset, the number of
indirect paths becomes substantial, and study-level summaries provide a
more practical and interpretable representation. Thus, the framework
naturally adapts to network size, supporting both fine-grained and
aggregated views of evidence contribution.

\rev{A defining feature of CSP is numerical and algebraic reproducibility with respect to the fitted NMA model. CSP does not produce a competing set of network treatment-effect estimates: for every target comparison, the canonical direct, indirect, study-level, and path-level contributions sum exactly to the single fitted NMA estimate obtained from the corresponding GLS model. For the fixed-effects diagnostics reported in the empirical applications, the diagnostic decomposition does not introduce a competing global statistic: the established generalized Cochran components satisfy \(Q_{\mathrm{net}}=Q_{\mathrm{het}}+Q_{\mathrm{inc}}\) exactly. Under a random-effects specification, heterogeneity is instead modeled through \(\tau^2\), and conditional on \(\mathbf V(\widehat{\tau}^2)\) the analogous algebraic identity is written \(Q_{\mathrm{net}}^{\mathrm{RE}}=Q_{\mathrm{error}}+Q_{\mathrm{inc}}^{\mathrm{RE}}\); these random-effects quantities are descriptive conditional diagnostics rather than the fixed-effects chi-squared tests. Thus, the NMA estimates, contribution displays, and corresponding \(Q\) diagnostics remain mutually compatible and traceable to the same observed contrasts and fitted covariance structure. Applied users who require exact reconstruction of the fitted estimate or the associated study-level forest, path, tension, and diagnostic displays may therefore use CSP without obtaining a second, potentially conflicting set of NMA treatment-effect estimates. For the fixed-effects decomposition, \(Q_{\mathrm{het}}\) and \(Q_{\mathrm{inc}}\) should nevertheless be interpreted jointly because their distinction is design-based; in particular, deviations in sparsely replicated designs may contribute to \(Q_{\mathrm{inc}}\) rather than to \(Q_{\mathrm{het}}\).}

Beyond the tools presented here, the CSP framework supports a broader set
of methodological and diagnostic extensions to the
network setting in a mathematically coherent and covariance-aware manner. In particular, the projection formulation
can be extended to network meta-regression by incorporating covariates into the
contrast-space design. Additional
developments, including tools for identifying influential paths, assessing sensitivity
to subsets of studies, and quantifying structural features of the network, are beyond
the scope of the present paper due to space constraints and will be presented in
future work.

\subsection*{\rev{Scope and limitations}}
\rev{The study-based direct and indirect contributions are uniquely defined and invariant to equivalent within-study parameterizations. The further decomposition of the total indirect contribution into individual study-labeled paths is optional and is primarily useful for interpretation and visualization. In general, a network flow can admit more than one valid decomposition into individual paths, so the individual path allocation is not uniquely determined by the total indirect contribution alone. When path-level reporting is desired, CSP therefore applies the deterministic path-extraction rule described in Section~\ref{sec:canonical_decomposition} to select one reproducible decomposition that is unique under the stated construction; users interested only in the total indirect contribution need not perform or report this finer path decomposition. The complete acyclic source-to-sink path representation is guaranteed under the arm-based covariance condition in Proposition~\ref{prop:acyclic_canonical_graph}, whereas the invariant study-level direct/indirect decomposition itself does not require this stronger condition. Under random effects, contribution weights and uncertainty measures are conditional on \(\mathbf V(\widehat{\tau}^2)\) and do not propagate uncertainty in \(\widehat{\tau}^2\). Finally, detailed path-level output can become cumbersome in dense networks, for which aggregated study-level summaries may be more practical.}

\paragraph{Acknowledgments}
Dr. O’Connor's research is supported by the Albert C. Dehn and Lois E. Dehn Chair in Veterinary Medicine, College of Veterinary Medicine, Michigan State University.
The authors used ChatGPT (OpenAI; GPT-5.6 Sol), accessed through \url{https://chatgpt.com} during manuscript preparation and revision in 2026, to assist with language editing, organization, clarity of presentation, and checking the consistency of mathematical notation and exposition. All AI-assisted suggestions were critically reviewed and verified by the authors, who take full responsibility for the accuracy, integrity, and originality of the manuscript. The methodology, statistical analyses, code, reported results, and scientific conclusions were developed, conducted, and verified by the authors.
\paragraph{Funding Statement}
The authors were supported by the National Science Foundation (NSF) under grant DMS-2413834, by the National Institutes of Health (NIH) under grants U19 AG068054 and U19 AG063893, and by the U.S. Department of Agriculture (USDA) National Institute of Food and Agriculture (NIFA) under grants IOWW-2022-10059 and 2024-68014-42361.
\paragraph{Ethics Statement}
\rev{Ethics approval was not required because this methodological study used only previously published, publicly available aggregate data and involved no new collection of human or animal participant data.}
\paragraph{Author Contributions}
\rev{Chong Wang: Conceptualization, Methodology, Project administration, Software, Supervision, Validation, Visualization, Writing--original draft, Writing--review \& editing. Yanqi Zhang: Investigation, Methodology, Software, Visualization, Writing--original draft, Writing--review \& editing. Zhezhen Jin: Methodology, Writing--review \& editing. Annette O'Connor: Investigation, Methodology, Validation, Visualization, Writing--review \& editing.}

\paragraph{Competing Interests}

The authors declare none.

\paragraph{Data Availability Statement}

Code implementing the contrast-space projection algorithm proposed in this paper is publicly available at \url{https://github.com/chongwangstat/CSP-NMA}. The repository includes an illustrative example of the implementation.

The three empirical datasets analyzed in the manuscript were obtained from publicly available online sources, as described in the corresponding references.

No new data were generated for this study.
\printbibliography

@article{White2012, author = {White, Ian R. and Barrett, Jessica K. and Jackson, Dan and Higgins, Julian P. T.}, title = {Consistency and inconsistency in network meta-analysis: model estimation using multivariate meta-regression}, journaltitle = {Research Synthesis Methods}, year = {2012}, volume = {3}, pages = {111--125}, doi = {10.1002/jrsm.1045}}

@article{Krahn2013, author = {Krahn, Ulrike and Binder, Harald and K{\"o}nig, Jochem}, title = {A graphical tool for locating inconsistency in network meta-analyses}, journaltitle = {BMC Medical Research Methodology}, year = {2013}, volume = {13}, pages = {35}, doi = {10.1186/1471-2288-13-35}}

@article{lu2011linear,
  author = {Lu, Guobing and Welton, Nicky J. and Higgins, Julian P. T. and White, Ian R. and Ades, A. E.},
  title = {Linear inference for mixed treatment comparison meta-analysis: A two-stage approach},
  journal = {Research Synthesis Methods},
  volume = {2},
  number = {1},
  pages = {43--60},
  year = {2011},
  doi = {10.1002/jrsm.34}
}

@article{Lu2004,
  author = {Lu, Gerta and Ades, A. E.},
  title = {Combination of direct and indirect evidence in mixed treatment comparisons},
  journal = {Statistics in Medicine},
  year = {2004},
  volume = {23},
  number = {20},
  pages = {3105--3124},
  doi = {10.1002/sim.1875}
}

@article{Dias2013,
  author  = {Dias, Sofia and Sutton, Alex J. and Ades, Anthony E. and Welton, Nicky J.},
  title   = {Evidence synthesis for decision making 2: A generalized linear modeling framework for pairwise and network meta-analysis of randomized controlled trials},
  journal = {Medical Decision Making},
  year    = {2013},
  volume  = {33},
  number  = {5},
  pages   = {607--617},
  doi     = {10.1177/0272989X12458724}
}

@article{Salanti2011,
  author  = {Salanti, Georgia and Ades, Anthony E. and Ioannidis, John P. A.},
  title   = {Graphical methods and numerical summaries for presenting results from multiple-treatment meta-analysis: An overview and tutorial},
  journal = {Journal of Clinical Epidemiology},
  year    = {2011},
  volume  = {64},
  number  = {2},
  pages   = {163--171},
  doi     = {10.1016/j.jclinepi.2010.03.016}
}

@article{RuckerSchwarzer2015,
  author  = {R{\"u}cker, Gerta and Schwarzer, Guido},
  title   = {Ranking treatments in frequentist network meta-analysis works without resampling methods},
  journal = {BMC Medical Research Methodology},
  year    = {2015},
  volume  = {15},
  pages   = {58},
  doi     = {10.1186/s12874-015-0060-8}
}

@book{hedges1985statistical,
  author    = {Hedges, Larry V. and Olkin, Ingram},
  title     = {Statistical Methods for Meta-Analysis},
  publisher = {Academic Press},
  year      = {1985},
  address   = {Orlando, FL}
}

@article{Hutton2015,
  author  = {Hutton, Brian and Salanti, Georgia and Caldwell, Deborah M. and others},
  title   = {The PRISMA extension statement for reporting of systematic reviews incorporating network meta-analyses of health care interventions: checklist and explanations},
  journal = {Annals of Internal Medicine},
  year    = {2015},
  volume  = {162},
  number  = {11},
  pages   = {777--784},
  doi     = {10.7326/M14-2385}
}

@article{Rucker2012,
  author  = {R{\"u}cker, Gerta},
  title   = {Network meta-analysis, electrical networks and graph theory},
  journal = {Research Synthesis Methods},
  year    = {2012},
  volume  = {3},
  number  = {4},
  pages   = {312--324},
  doi     = {10.1002/jrsm.1058}
}

@article{Koenig2013,
  author  = {K{\"o}nig, Johannes and Krahn, Uta and Binder, Harald},
  title   = {Visualizing the flow of evidence in network meta-analysis and characterizing mixed treatment comparisons},
  journal = {Statistics in Medicine},
  year    = {2013},
  volume  = {32},
  number  = {30},
  pages   = {5414--5429}
}

@article{Papakonstantinou2018,
  author  = {Papakonstantinou, Theodoros and Nikolakopoulou, Adriani and Egger, Matthias and Salanti, Georgia},
  title   = {Estimating the contribution of studies and specific outcomes to network meta-analysis results},
  journal = {Statistics in Medicine},
  year    = {2018},
  volume  = {37},
  number  = {12},
  pages   = {1947--1960}
}

@article{Davies2022,
  author  = {Davies, Annabel L. and Papakonstantinou, Theodoros and Nikolakopoulou, Adriani and R{\"u}cker, Gerta and Galla, Tobias},
  title   = {Network meta-analysis and random walks},
  journal = {Statistics in Medicine},
  year    = {2022},
  volume  = {41},
  number  = {11},
  pages   = {2091--2114},
  doi     = {10.1002/sim.9346}
}

@article{Rucker2020,
  author  = {R{\"u}cker, Gerta and Nikolakopoulou, Adriani and Papakonstantinou, Theodoros and Salanti, Georgia and Riley, Richard D. and Schwarzer, Guido},
  title   = {The statistical importance of a study for a network meta-analysis estimate},
  journal = {BMC Medical Research Methodology},
  year    = {2020},
  volume  = {20},
  pages   = {190}
}

@article{Mao2025,
  author  = {Mao, Y. and Shen, Y. and Yang, Q. and Shi, Q. and Li, S.},
  title   = {A leave-one-out algorithm for contribution analysis in component network meta-analysis},
  journal = {BMC Medical Research Methodology},
  year    = {2025},
  volume  = {25},
  pages   = {191}
}

@article{Rucker2014reduce,
  author  = {R{\"u}cker, Gerta and Schwarzer, Guido},
  title   = {Reduce dimension or reduce weights? Comparing two approaches to multi-arm studies in network meta-analysis},
  journal = {Statistics in Medicine},
  year    = {2014},
  volume  = {33},
  number  = {25},
  pages   = {4353--4369},
  doi     = {10.1002/sim.6236}
}

@article{Jackson2012,
  author  = {Jackson, Dan and White, Ian R. and Riley, Richard D.},
  title   = {Quantifying the impact of between-study heterogeneity in multivariate meta-analysis},
  journal = {Statistics in Medicine},
  year    = {2012},
  volume  = {31},
  number  = {29},
  pages   = {3805--3820},
  doi     = {10.1002/sim.5453}
}

@article{Higgins2012,
  author  = {Higgins, Julian P. T. and Jackson, Dan and Barrett, James K. and Lu, Guobing and Ades, A. E. and White, Ian R.},
  title   = {Consistency and inconsistency in network meta-analysis: concepts and models for multi-arm studies},
  journal = {Research Synthesis Methods},
  year    = {2012},
  volume  = {3},
  number  = {2},
  pages   = {98--110}
}

@article{siemieniuk2020drug,
  author  = {Siemieniuk, Reed A. C. and others},
  title   = {Drug treatments for COVID-19: living systematic review and network meta-analysis},
  journal = {British Medical Journal},
  year    = {2020},
  volume  = {370},
  pages   = {m2980},
  doi     = {10.1136/bmj.m2980}
}

@article{WHO_REACT_2020,
  author  = {{WHO Rapid Evidence Appraisal for COVID-19 Therapies (REACT) Working Group} and Sterne, Jonathan A. C. and others},
  title   = {Association between administration of systemic corticosteroids and mortality among critically ill patients with COVID-19: a meta-analysis},
  journal = {Journal of the American Medical Association},
  year    = {2020},
  volume  = {324},
  number  = {13},
  pages   = {1330--1341}
}

@article{WHO_Guideline_2020,
  author  = {Lamontagne, Fran{\c{c}}ois and others},
  title   = {A living WHO guideline on drugs for COVID-19},
  journal = {British Medical Journal},
  year    = {2020},
  volume  = {370},
  pages   = {m3379}
}

@article{Sbidian2021,
  author  = {Sbidian, Emilie and others},
  title   = {Systemic pharmacological treatments for chronic plaque psoriasis: a network meta-analysis},
  journal = {Cochrane Database of Systematic Reviews},
  year    = {2021},
  number  = {4}
}

@article{Cipriani2018,
  author  = {Cipriani, Andrea and others},
  title   = {Comparative efficacy and acceptability of 21 antidepressant drugs for the acute treatment of adults with major depressive disorder: a systematic review and network meta-analysis},
  journal = {The Lancet},
  year    = {2018},
  volume  = {391},
  number  = {10128},
  pages   = {1357--1366}
}

@article{White2015,
  author  = {White, Ian R.},
  title   = {Network meta-analysis},
  journal = {The Stata Journal},
  year    = {2015},
  volume  = {15},
  number  = {4},
  pages   = {951--985},
  doi     = {10.1177/1536867X1501500403}
}

@article{ShihTu2021,
  author  = {Shih, Ming-Chieh and Tu, Yu-Kang},
  title   = {An evidence-splitting approach to evaluation of direct-indirect evidence inconsistency in network meta-analysis},
  journal = {Research Synthesis Methods},
  year    = {2021},
  volume  = {12},
  number  = {2},
  pages   = {226--238},
  doi     = {10.1002/jrsm.1480}
}

\end{document}